\documentclass[11pt]{article}

\usepackage[T1]{fontenc}
\usepackage{lmodern}
\usepackage[margin=1in]{geometry}
\usepackage{amsmath,amssymb,amsthm,mathtools,mathrsfs}
\usepackage{booktabs}
\usepackage{enumitem}
\usepackage{microtype}
\usepackage{array}
\usepackage{float}
\usepackage{tikz}
\usepackage{pgfplots}
\usepackage[most]{tcolorbox}
\usepackage{url}
\usepackage{natbib}
\usepackage[colorlinks=true,linkcolor=blue,citecolor=blue,urlcolor=blue,
  pdftitle={Harmonic Ranking: A 0.698-Competitive Algorithm for Edge-Weighted Oblivious Matching},
  pdfauthor={Anonymous Author(s)},
  pdfsubject={Edge-weighted oblivious bipartite matching},
  pdfkeywords={oblivious matching, Ranking, randomized greedy, minimum cut, maximum flow, factor-revealing program}]{hyperref}
\pgfplotsset{compat=1.14}

\definecolor{thresholdblue}{RGB}{48,105,178}
\definecolor{thresholdred}{RGB}{213,92,67}
\definecolor{sharedviolet}{RGB}{111,82,154}

\definecolor{flowblue}{RGB}{43,101,176}
\definecolor{flowred}{RGB}{205,82,60}

\allowdisplaybreaks
\setlist[itemize]{leftmargin=1.5em}
\setlist[enumerate]{leftmargin=1.8em}
\newtheorem{theorem}{Theorem}[section]
\newtheorem{proposition}[theorem]{Proposition}
\newtheorem{lemma}[theorem]{Lemma}
\newtheorem{corollary}[theorem]{Corollary}
\theoremstyle{definition}
\newtheorem{definition}[theorem]{Definition}
\newtheorem{example}[theorem]{Example}
\theoremstyle{remark}
\newtheorem{remark}[theorem]{Remark}

\newcommand{\R}{\mathbb R}
\newcommand{\E}{\mathbb E}
\newcommand{\one}{\mathbf 1}
\newcommand{\dd}{\,\mathrm d}
\newcommand{\ALG}{\operatorname{ALG}}
\newcommand{\OPT}{\operatorname{OPT}}
\newcommand{\MC}{\operatorname{MC}}
\newcommand{\MF}{\operatorname{MF}}
\newcommand{\polyLPp}{\operatorname{polyLP}'}
\newcommand{\Unif}{\operatorname{Unif}}
\newcommand{\HR}{\textsc{Harmonic Ranking}}

\newtcolorbox{algorithmbox}[1]{
  enhanced,
  breakable,
  colback=black!2,
  colframe=black!55,
  boxrule=.55pt,
  arc=1pt,
  left=7pt,
  right=7pt,
  top=5pt,
  bottom=5pt,
  title={#1},
  fonttitle=\bfseries,
  coltitle=black,
  attach boxed title to top left={xshift=6pt,yshift=-2mm},
  boxed title style={colback=white,colframe=black!55,boxrule=.55pt,arc=1pt}
}

\title{Harmonic Ranking for Edge-Weighted Oblivious Matching}

\author{
Bo Peng\thanks{
Email: \texttt{ahqspbo@gmail.com}
}
\qquad
Zhihao Gavin Tang\thanks{
Email: \texttt{tang.zhihao@mail.shufe.edu.cn}
}\\[1mm]
Shanghai University of Finance and Economics
}
\date{}

\begin{document}
\maketitle

\begin{abstract}
We study edge-weighted oblivious bipartite matching.  The weight of
every potential edge is known, but its existence is revealed only when
the edge is probed, and a successful probe between two free vertices
must be accepted immediately.  We give an explicit randomized
algorithm with certified competitive ratio \(0.698\), improving
the previous best guarantee of \(0.659\) (Huang, Sun, Wu, and Zhao,
FOCS 2025).  The result is computer-assisted and is verified by a
reproducible exact-integer computation.  The same algorithm has
a \(0.698\)-competitive online implementation for the vertex-weighted
random-arrival model, improving the previous \(0.696\) unweighted
guarantee of Mahdian and Yan
(STOC 2011) and the \(0.686\) vertex-weighted guarantee of Peng and Tang
(EC 2025).

Our algorithm, \HR, is a role-symmetric generalization of
\textsc{Ranking}.  It assigns an independent random rank \(x_z\) to
each vertex and probes a potential edge \(uv\) in decreasing order of
\[
  w_{uv}\frac{h(x_u)h(x_v)}{h(x_u)+h(x_v)}.
\]
This harmonic priority arises from a budget-balanced gain split and a
mutual-proposal interpretation.
The analysis lifts two cutoff curves into indicators, reducing the
exponential-size factor-revealing problem to a polynomial-size directed
minimum-cut instance.  A maximum-flow computation with rounded-down
integer capacities gives a rigorous certificate. 
Independently, we observe that the finite-grid unweighted relaxation of
our factor-revealing program coincides exactly with a Mahdian--Yan program.
\end{abstract}

\medskip
\noindent\textbf{Keywords:}
oblivious matching, Ranking, randomized greedy, min-cut, max-flow,
factor-revealing program.

\newpage

\section{Introduction}

\paragraph{Oblivious Bipartite Matching.} 
We study oblivious matching on \emph{edge-weighted bipartite graphs}. The algorithm knows the bipartition $(L,R)$ and the nonnegative weight of every potential edge in $(L\times R)$, but not which potential edges are realized. Using only these weights and its private randomness, the algorithm determines an order in which to probe them. Whenever a probe reveals a realized edge whose endpoints are both free, the algorithm must accept it immediately and irrevocably. Thus probing is not a preliminary information-gathering phase: the probing order itself is the algorithm.

The performance of an algorithm is measured by its competitive ratio: the largest \(\Gamma\) such that, on every instance, the expected weight of the algorithm's matching is at least \(\Gamma\) times the weight of a maximum-weight matching.  The deterministic greedy algorithm, which probes potential edges in nonincreasing order of weight, is \(1/2\)-competitive, and no deterministic algorithm can achieve a better ratio.

A recurring principle behind successful randomized greedy matching algorithms is to attach randomness to vertices and let the resulting vertex ranks induce the probing order. The challenge in the edge-weighted setting is how to combine the ranks of the two endpoints.
\citet{TangWuZhang2023} use the rank of only one endpoint, resulting in a one-sided rule for a symmetric problem. \citet{HuangSunWuZhao2025} restore symmetry through a product of the two endpoint scores, but leave the probing priority and the gain-sharing rule as separate design choices. 
Our guiding principle is that the gain-sharing rule itself should determine the probing order. This leads to \(\HR\), the following symmetric generalization of \textsc{Ranking}.

\begin{algorithmbox}{Harmonic Ranking with score function \(h\)}
Fix a nonincreasing, right-continuous function
\(h:[0,1)\to\R_{>0}\).
Independently sample rank \(x_z\sim\Unif[0,1)\) for every vertex \(z \in L \cup R\). 
Give every potential pair \(uv\) priority
\begin{equation}\label{eq:intro-priority}
  \pi_{uv}(\boldsymbol x)
  :=
  \frac{h(x_u)h(x_v)}{h(x_u)+h(x_v)}\,w_{uv}.
\end{equation}
Probe the potential pairs in descending priority, breaking ties by
any admissible rule in Section~\ref{sec:tie-breaking}.
\end{algorithmbox}

\begin{theorem}\label{thm:main}
There exists an explicit score function \(h^\star\) such that for every
edge-weighted oblivious instance \(G\), \(\HR(h^\star)\) satisfies
\[
  \E[\ALG(G)]
  \ge
  0.698015475248\,\OPT(G).
\]
\end{theorem}

This improves on the \(1-1/e\) guarantee of the one-sided
edge-weighted algorithm of \citep{TangWuZhang2023}, and on the \(0.659\) guarantee of the two-sided algorithm of \citep{HuangSunWuZhao2025}.

\newpage
\paragraph{Online consequence for vertex-weighted matching.}

The \textsc{Ranking} algorithm was originally introduced by \citet{KarpVaziraniVazirani1990} for online
bipartite matching.  In the vertex-weighted random-arrival model,
\(\HR\) has a standard online implementation: process arrivals in
order and maximize the harmonic priority among the arriving vertex's
available neighbors.  If equal global priorities are ordered by
arrival time, this execution is pathwise equivalent to the global
harmonic-priority scan.
Section~\ref{sec:online} gives the coupling argument.

\nopagebreak[4]
\begin{corollary}\label{cor:online}
Assuming that the number of online vertices is known in advance,
Online Harmonic Ranking is \(0.698015475248\)-competitive for
vertex-weighted bipartite matching with random arrivals.
\end{corollary}
This simultaneously improves the \(0.696\) guarantee for unweighted
graphs~\citep{MahdianYan2011} and the \(0.686\) guarantee for
vertex-weighted graphs~\citep{PengTang2025}.

\subsection{Technical overview}

\paragraph{Harmonic Ranking from score-balanced sharing.}
The design of \(\HR\) is not ad hoc. It emerges from the randomized
primal--dual framework and its economic interpretation through gain
sharing. In online bipartite matching, offline vertices are viewed as
items and online vertices as buyers. Each item posts a price determined
by its random rank and the current arrival time, and the arriving buyer
chooses the available neighboring item that offers the greatest
utility.

The main obstacle to extending this interpretation to edge-weighted
graphs is its inherent asymmetry. An edge weight belongs to a pair,
rather than to either endpoint, so there is no canonical choice of
which endpoint should be the item and which should be the buyer. We
remove this asymmetry through \emph{mutual proposals}. We first specify
how the weight of a potential pair \(uv\) would be divided between
\(u\) and \(v\) if that pair were selected. Each free vertex then
proposes along the available incident pair that offers it the largest
prospective gain. The algorithm probes a pair precisely when its two
endpoints propose to each other.

For this procedure to be well defined, a mutual proposal must exist
whenever an available pair remains. A general gain-sharing rule need
not satisfy this requirement: the proposals may instead form a
directed cycle of length greater than two. We show that, among positive
complementary rank-based gain-sharing rules, ruling out such cycles on
every instance uniquely forces score-balanced sharing, up to rescaling
its underlying score function. Under this rule, multiplying each
endpoint's prospective gain by its own score produces the same
harmonic edge priority. The formal characterization and its proof
appear in Section~\ref{sec:model}.

\paragraph{From the variational formulation to min-cut and max-flow.}
We first establish edge-weighted analogues of the two structural
ingredients underlying analyses of \textsc{Ranking}: a cutoff lemma and
an alternating-path insertion lemma. For \(\HR\), these results yield a
two-curve variational lower bound in the spirit of
\citep{HuangTangWuZhang2019,JinWilliamson2022,PengTang2025}.

A direct grid discretization that enumerates the possible curve pairs
has exponential size, making fine-grid factor-revealing optimization
infeasible. Our \(0.698\) guarantee is enabled by a new lifting: we
represent each one-dimensional threshold curve \(a\) by its
two-dimensional indicator
\[
    U_a(x,y):=\mathbf 1\{y>a(x)\}
\]
on the unit square. Under this representation, the threshold
conditions become directed precedence constraints.  Once the score is
fixed to be constant on a uniform grid, they produce a polynomial-size
directed network whose minimum-cut value lower-bounds the continuous
variational objective.  The maximum-flow problem on the same network
is its strong dual.  Numerical optimization is used only to choose a
candidate score; after that choice is fixed, a reproducible
exact-integer maximum-flow computation verifies the stated guarantee.

A second conceptual contribution of our framework is a complementary
variational hierarchy for the edge-weighted, vertex-weighted, and
unweighted analyses.  At the two-monotone unweighted level, the
complementary exact whole-cell game has the same optimum value as the polynomial-size program of \citep{MahdianYan2011}.

\subsection{Further related work}

\paragraph{Stochastic query--commit matching.}
In the stochastic query--commit model, each potential edge \(e\) has
a known weight \(w_e\) and realization probability \(p_e\).  The
standard formulation assumes independent edge realizations, permits
adaptive probes, and requires every successful probe to be committed
immediately.  Early formulations and commitment variants were studied
by \citet{ChenImmorlicaKarlinMahdianRudra2009} and \citet{CostelloTetaliTripathi2012}. \citet{GamlathKaleSvensson2019} gave a
\((1-1/e)\)-approximation for edge-weighted bipartite graphs.
\citet{DerakhshanFarhadi2023} improved this to
\(1-1/e+\delta\) for an absolute constant \(\delta>0.0014\), and
\citet{ChenHuangLiTang2025} subsequently obtained a
\(0.641\)-approximation. For the unweighted and vertex-weighted cases, which are closely
connected to online stochastic matching, the state-of-the-art
\(0.705\)-approximation for query--commit matching is due to
\citet{ChenHuangSun2024}; see their paper for the related literature.

\citet{TangWuZhang2023} observed that a guarantee holding for every
fixed realized graph transfers by averaging to stochastic instances,
even under arbitrarily correlated realizations.  The same argument
shows that Theorem~\ref{thm:main} gives a \(0.698\)-approximation for
edge-weighted stochastic query--commit matching on bipartite graphs
against the expected omniscient optimum, while using neither the
probabilities \(p_e\) nor adaptive probing.

\paragraph{Online bipartite matching.}
The randomized primal--dual analysis, also known as the gain sharing framework, was developed by \citet{DevanurJainKleinberg2013}. It provides a unified analysis of the optimal \(1-1/e\) competitive ratio for online bipartite matching on unweighted~\cite{KarpVaziraniVazirani1990} and vertex-weighted graphs~\cite{AggarwalGKM11} under adversarial arrivals.

Random arrivals permit stronger guarantees.  In the unweighted setting,
\citet{KarandeMehtaTripathi2011} obtained \(0.653\), and \citet{MahdianYan2011} obtained \(0.696\). 
For vertex-weighted matching, the primal--dual analyses of \citep{HuangTangWuZhang2019} and \citep{JinWilliamson2022} were further sharpened by \citet{PengTang2025}, who obtained a \(0.686\) guarantee.
Our online consequence belongs to this latter vertex-weighted random-arrival
setting: Section~\ref{sec:online} connects the standard online rule
directly to Harmonic Ranking through an arrival-order coupling.

\paragraph{Ranking on general graphs.}
The \textsc{Ranking} algorithm extends naturally to general, non-bipartite
graphs. A sequence of works has progressively improved its competitive
ratio to \(0.560\)~\citep{ChanChenWuZhao2018,
DerakhshanRoghaniSaneianYu2026, DerakhshanYu2026}.
In a different direction, \citet{HuangKangTangWuZhangZhu2020} adapted
\textsc{Ranking} to the fully online matching model, where vertices depart
in an adversarial order and matching decisions are made upon departure,
and established a \(0.521\) competitive ratio on general graphs. This
guarantee was subsequently improved to \(0.539\) by
\citet{DerakhshanYu2026}.

\paragraph{Randomized greedy rules beyond \textsc{Ranking}.}
\citet{DyerFrieze1991} studied the uniformly random-edge greedy rule,
whose worst-case guarantee on general graphs does not exceed \(1/2\).
\citet{AronsonDyerFriezeSuen1995} introduced
\textsc{Modified Randomized Greedy} (\textsc{MRG}), which samples a
uniformly random decision order together with independent uniformly
random preference orders, and proved a guarantee strictly above
\(1/2\).  \citet{TangWuZhang2023} analyzed the weaker
\textsc{Random Decision Order} (\textsc{RDO}) rule, which randomizes
only the decision order while allowing arbitrary fixed preference
orders.  They obtained guarantees of \(0.639\) on bipartite graphs
and \(0.531\) on general graphs; the latter immediately gives the same
\(0.531\) guarantee for \textsc{MRG}.

\subsection{Roadmap}

Section~\ref{sec:model} derives \(\HR\) from gain sharing and mutual proposals. The score-balanced characterization explains the design of the algorithm but is not used in the competitive-ratio proof.
Section~\ref{sec:structure} establishes the cutoff and alternating-path insertion lemmas and highlights the absence of cutoff-curve monotonicity.
Section~\ref{sec:game} develops the two-curve variational game.
Section~\ref{sec:finite} then restricts to uniform-grid step functions, derives an exact minimum cut for the finite cellwise game, and writes its strongly dual maximum-flow program.
Section~\ref{sec:certificate} constructs a \(240\)-step score and verifies the \(0.698\) guarantee by reproducible exact-integer computation.  The proof of Theorem~\ref{thm:main} is then complete. 
Section~\ref{sec:online} derives Corollary~\ref{cor:online} by a direct
coupling with Online Harmonic Ranking.
Section~\ref{sec:correspondences} states an independent unweighted-grid correspondence, whose details are moved to Appendix~\ref{app:grid-correspondence}. 
Finally, Section~\ref{sec:discussion} discusses open questions.

\section{Score-Balanced Sharing and Harmonic Ranking}\label{sec:model}

We first formalize the model and the primal-dual gain-sharing framework, then characterize the score-balanced rule and establish the equivalence between Mutual Proposals and Harmonic Ranking.

\subsection{Model and notation}

An edge-weighted oblivious instance is a tuple
\[
 G=(L,R,w,E),
 \qquad
 w=(w_{uv})_{uv\in L\times R}\in\R_{\ge0}^{L\times R},
 \qquad
 E\subseteq L\times R,
\]
where \(L\) and \(R\) are disjoint finite vertex sets. 
The instance \(G\) is fixed before the algorithm samples its private
randomness.
The algorithm is given \(L,R,w\), but not \(E\); probing \(uv\) reveals whether \(uv\in E\). 
When a probe reveals an edge and both endpoints are free, the pair must be
added to the matching. The reward is the total weight selected.

Define
\[
 \OPT(G)
 :=
 \max\left\{
   \sum_{uv\in M}w_{uv}:
   M\subseteq E\text{ is a matching}
 \right\},
\]
and let \(\ALG(G)\) denote the algorithm's random reward.  All
probabilities and expectations are over the algorithm's private
randomness with \(G\) fixed.  When the instance is fixed, we abbreviate
\(\ALG:=\ALG(G)\) and \(\OPT:=\OPT(G)\).  Throughout, ``pair'' means a
potential pair in \(L\times R\), whereas ``edge'' means a realized
pair in \(E\).





\subsection{Primal--dual gain sharing}

The standard fractional matching relaxation and its fractional
vertex-cover dual are
\begin{equation}\label{eq:matching-LP-pair}
\begin{array}{lll@{\qquad\qquad}lll}
 \text{\rm(P)}&
 \displaystyle\max&\displaystyle\sum_{uv\in E}w_{uv}x_{uv}
 &
 \text{\rm(D)}&
 \displaystyle\min&\displaystyle\sum_{z\in L\cup R}\alpha_z\\[4pt]
 &\text{\rm s.t.}&\displaystyle\sum_{e\ni z}x_e\le1
       \quad(z\in L\cup R)
 &&
 \text{\rm s.t.}&\alpha_u+\alpha_v\ge w_{uv}
       \quad(uv\in E),\\
 &&x_e\ge0\quad(e\in E)
 &&&\alpha_z\ge0\quad(z\in L\cup R).
\end{array}
\end{equation}
A gain-sharing rule specifies how the weight of a selected edge is
credited to its endpoints.  If the credits are nonnegative and sum to
the selected weight, then \(\sum_z\alpha_z=\ALG\) in every run.  It is
therefore enough to establish, for every edge $uv \in E$,
\[
  \E[\alpha_u+\alpha_v]\ge \Gamma w_{uv}.
\]
The scaled expectations \(\{\E[\alpha_z]/\Gamma\}_z\) are then dual
feasible, and weak duality gives
\[
  \E[\ALG]=\sum_z\E[\alpha_z]\ge\Gamma\,\OPT.
\]

This is usually presented as an analysis of a given algorithm.  Here
we also use it as an algorithm-design principle: first decide how the
weight of a hypothetical selected pair would be shared, and then let
the two prospective shares determine which pair should be probed.

\subsection{Mutual Proposals}
\label{sec:mutual-proposals}

Fix a measurable gain-sharing rule
\[
 g:[0,1)^2\to (0,1)
 \qquad\text{such that}\qquad
 g(x,y)+g(y,x)=1
 \quad(x,y\in[0,1)).
\]
Every vertex \(z\) independently samples a rank
\(x_z\sim\Unif[0,1)\).  On a potential pair \(uv\), define the two
prospective gains
\begin{equation}\label{eq:general-prospective-gains}
  G_u(uv):=g(x_u,x_v)w_{uv},
  \qquad
  G_v(uv):=g(x_v,x_u)w_{uv}.
\end{equation}
They depend only on the known weight and sampled ranks, not on whether
the pair is realized, and complementarity gives
\(G_u(uv)+G_v(uv)=w_{uv}\).

\begin{algorithmbox}{Mutual Proposals with gain-sharing rule \(g\)}
Initially every vertex is free and every potential pair is unprobed.
Call an unprobed pair with two free endpoints \emph{available}.

At each stage, every free vertex incident to an available pair
proposes along one that maximizes its own prospective gain in
\eqref{eq:general-prospective-gains}.  Resolve all ties using the same
common total order on potential pairs.  If one or more pairs are
proposed by both endpoints, probe the first such pair in that order.
A failed probe marks the pair as probed; a successful probe
matches its endpoints.  Then recompute all proposals.  If available
pairs remain but no pair is proposed by both endpoints, the procedure
is stuck.
\end{algorithmbox}

When a successful probe selects \(uv\), assign the analytical credits
\(\alpha_u:=G_u(uv)\) and \(\alpha_v:=G_v(uv)\).  Thus every selected
weight is split exactly between its endpoints, as required by the
primal--dual argument.

At any stage, orient an arc from each proposing vertex toward the
other endpoint of its proposal.  Every nonisolated vertex has
out-degree one, so each nontrivial component of this proposal graph
contains a directed cycle.  Because the graph is bipartite, such a
cycle is either a mutual two-cycle or has length at least four.  We
call the latter a \emph{long proposal cycle}, and call \(g\)
\emph{long-cycle-free on the sampled domain} if no such cycle can
arise in any finite weighted bipartite residual instance, under any
assignment of ranks in \([0,1)\), and with any common total tie order.
If \(g\) is long-cycle-free, Mutual Proposals is
well-defined: whenever an available pair remains, some component
contains a two-cycle and hence a mutually proposed pair.  Each probe
then either deletes one pair or matches two vertices, so the procedure
terminates after finitely many probes.


\begin{theorem}
\label{thm:cycle-characterization}
Let \(g:[0,1)^2\to(0,1)\) satisfy
\[
 g(x,y)+g(y,x)=1
 \qquad(x,y\in[0,1)).
\]
The following statements are equivalent.
\begin{enumerate}
 \item The rule \(g\) is long-cycle-free on the sampled domain.
 \item There is a function \(h:[0,1)\to\R_{>0}\) such that
 \[
  g(x,y)=\frac{h(y)}{h(x)+h(y)}
  \qquad(x,y\in[0,1)).
 \]
\end{enumerate}
The function \(h\) is unique up to multiplication by a positive
constant. 
\end{theorem}

\begin{proof}
Define the positive odds ratio
\[
 O(x,y):=\frac{g(x,y)}{g(y,x)}.
\]
Consider a four-cycle whose consecutive vertices have ranks
\(x_1,x_2,x_3,x_4\), and let \(w_i\) be the weight of the clockwise
edge from rank \(x_i\) to rank \(x_{i+1}\), with indices modulo four.
The vertices remain distinct when some ranks coincide.  Write
\[
 Q:=
 O(x_1,x_2)O(x_2,x_3)O(x_3,x_4)O(x_4,x_1).
\]
If \(Q>1\), choose the weights so that
\[
 \frac{w_i}{w_{i-1}}
 =Q^{1/4}\frac{g(x_i,x_{i-1})}{g(x_i,x_{i+1})}.
\]
The prescribed ratios multiply to one, and at every vertex
\[
 g(x_i,x_{i+1})w_i
 =Q^{1/4}g(x_i,x_{i-1})w_{i-1}
 >g(x_i,x_{i-1})w_{i-1}
\]
so the proposal graph is a strict clockwise long cycle.  If \(Q<1\),
the reverse orientation gives the same contradiction.  Thus
long-cycle-freeness forces \(Q=1\) for every four ranks.

Assign the consecutive ranks \(x,y,z,x\).  Since \(O(x,x)=1\), the
four-cycle identity becomes
\[
 O(x,y)O(y,z)=O(x,z).
\]
Fix a reference rank \(x_0\) and set \(h(x):=O(x_0,x)\).  The last
identity gives \(O(x,y)=h(y)/h(x)\).  Combining this with
complementarity that $g(x,y)+g(y,x)=1$ yields
\[
 g(x,y)=\frac{h(y)}{h(x)+h(y)}.
\]
Changing \(x_0\) rescales \(h\) by one positive constant.
If \(\widetilde h\) is any other representation, then
\[
 \frac{\widetilde h(y)}{\widetilde h(x)}
 =O(x,y)
 =\frac{h(y)}{h(x)},
\]
so \(\widetilde h=ch\) for some \(c>0\).

Conversely, suppose \(g\) has the form
\( g(x,y) = \frac{h(y)}{h(x)+h(y)}.\)
The local preference score of \(uv\) at \(u\), multiplied by the
vertex-dependent constant \(h(x_u)\), is the common priority
\[
 h(x_u)g(x_u,x_v)w_{uv}
 =\frac{h(x_u)h(x_v)}{h(x_u)+h(x_v)}w_{uv}.
\]
Around a directed proposal cycle, the priority of each outgoing edge
is at least that of the incoming edge.  One strict comparison gives an
impossible strict cyclic inequality.  If all priorities are equal,
each outgoing edge must precede the incoming edge in the common tie
order, again producing an impossible strict cycle.  Hence no long
proposal cycle exists.
\end{proof}

The characterization itself does not orient the rank scale.
Requiring a smaller partner rank to be weakly more attractive makes
\(h\) nonincreasing. 

From now on, we restrict to \emph{score functions} \(h:[0,1)\to\R_{>0}\) that are nonincreasing and right-continuous. We denote by $g_h$ the corresponding \emph{score-balanced} gain sharing rule:
\begin{equation}\label{eq:g-def}
 g_h:[0,1)^2\to(0,1),
 \qquad
 g_h(x,y):=\frac{h(y)}{h(x)+h(y)}.
\end{equation}

Finally, we establish that Mutual Proposals equals Harmonic Ranking.


\begin{lemma}
\label{lem:proposal-priority}
For every score function \(h\) and every rank realization, Mutual Proposals with
gain-sharing rule \(g_h\) is well-defined and produces exactly the
same matching as \(\HR(h)\).
\end{lemma}

\begin{proof}
For every available \(uv\),
\begin{equation}\label{eq:common-priority-identity}
 h(x_u)G_u(uv)=h(x_v)G_v(uv)=\pi_{uv}.
\end{equation}
Because the multiplying factor is fixed at each vertex, every vertex
orders its incident pairs by the common priority \(\pi\).  Thus the
highest-priority available pair is always proposed by both endpoints,
so Mutual Proposals never gets stuck.

Say that one pair \emph{precedes} another when it has larger numerical
priority, or when the priorities are equal and it is earlier in the tie-breaking
order.  When Mutual Proposals probes a pair \(e\), every
unprobed pair preceding \(e\) is either unavailable or vertex-disjoint
from \(e\); otherwise a shared endpoint would propose the preceding
pair.  Starting from the global scan, we may therefore move \(e\)
ahead of all such pairs: each exchange passes either a no-op or a
vertex-disjoint probe and preserves both acceptances.  Repeating this
exchange after each Mutual Proposals probe transforms the global scan
into the Mutual Proposals sequence.  Hence both select the same matching.
\end{proof}

In the remainder of the paper, we study the  \(\HR\) algorithm. For its analysis, when \(uv\) is selected, assign endpoint gains
\begin{equation}\label{eq:gains}
 \alpha_u
 :=
 \frac{h(x_v)}{h(x_u)+h(x_v)}w_{uv},
 \qquad
 \alpha_v
 :=
 \frac{h(x_u)}{h(x_u)+h(x_v)}w_{uv}.
\end{equation}
Unmatched vertices receive gain zero.  Then
\begin{equation}\label{eq:gain-identities}
\begin{aligned}
 \alpha_u+\alpha_v&=w_{uv},&
 \pi_{uv}&=h(x_u)\alpha_u=h(x_v)\alpha_v,\\
 \ALG&=\sum_{z\in L\cup R}\alpha_z.&&
\end{aligned}
\end{equation}

\subsection{Tie-breaking}
\label{sec:tie-breaking}
The mutual proposals algorithm resolves all ties using a single common
total order on potential pairs. This global consistency is crucial for
preventing long cycles: without such a requirement, a cycle of arbitrary
length could arise when all involved pairs have identical priorities.

However, we consider the  \(\HR\) algorithm and allow a broader class of tie-breaking rules. Specifically, we call a tie-breaking rule \emph{admissible} if the
resulting total scan order is a measurable function of the vertex ranks
and satisfies the following local monotonicity property: decreasing the
rank of a vertex \(z\) leaves the relative order of all pairs not
incident to \(z\) unchanged, and cannot move any pair incident to \(z\)
later relative to any pair not incident to \(z\).

As long as this property is satisfied, all structural properties
established in the next section continue to hold. The common total order
on potential pairs used by the mutual proposals algorithm is admissible,
as are various rank-dependent tie-breaking rules. In particular, the
tie-breaking convention introduced for the online specialization in
Section~\ref{sec:online} is admissible. Consequently, all competitive
guarantees proved for Harmonic Ranking below extend to the online setting.

\section{Cutoffs and alternating paths}\label{sec:structure}

The cutoff and vertex-insertion arguments in this section are standard ingredients in analyses of unweighted and vertex-weighted \textsc{Ranking}~\citep{MahdianYan2011,HuangTangWuZhang2019,JinWilliamson2022,PengTang2025}.  
The challenge in the edge-weighted setting is not that rank-perturbation techniques cease to apply, but that the additional structure introduced by edge weights limits how finely we can classify algorithmic outcomes while preserving monotonicity under a one-rank perturbation.

In classical online bipartite matching, condition on the arrival order
and on every offline rank except that of the offline endpoint \(u\) of
a target edge \(uv\).  As this last rank ranges over \([0,1)\), the
three outcomes---\(v\) is strictly earlier, \(uv\) is selected, and
\(u\) is strictly earlier---occupy three consecutive, possibly empty
intervals, up to boundary conventions.  The same offline-rank
trichotomy underlies the vertex-weighted analyses.  It need not
survive, however, when the varied coordinate is the arrival timestamp
of the online endpoint.  In the conventional vertex-weighted
formulation, this loss is reflected in the fact that the upper boundary
\(a\), viewed as a function of the online timestamp, need not be
nondecreasing.

The robust classification is necessarily coarser.  In the notation below, after fixing \(x_u\) and varying \(y=x_v\), we distinguish whether \(v\) is matched no later than \(u\), or whether \(u\) is strictly earlier than \(v\), with respect to the probing order.
The first class is downward closed in \(y\), even though runs
in which \(v\) is strictly earlier and runs in which \(uv\) itself is
selected may be interleaved within it.  Symmetrically, after fixing
\(x_v\) and varying \(x_u\), the class in which \(u\) is no later may
interleave the outcomes ``\(u\) is strictly earlier'' and ``\(uv\) is
selected.''  This binary classification restores a cutoff in every
one-dimensional section, but it does not force either cutoff curve to
be monotone across sections.

Figure~\ref{fig:threshold-monotonicity} summarizes the structural
distinction among the three matching models.  Fix a realized target
edge \(uv\) and all other ranks.  Informally, the rank square splits
into the region \(U\) where \(u\) is strictly earlier, the region \(V\)
where \(v\) is strictly earlier, and the region \(S\) where \(uv\) is
selected.  Every vertical section of \(U\) and every horizontal
section of \(V\) is an interval.  What differs across the models is
whether their boundary curves are nondecreasing across sections.

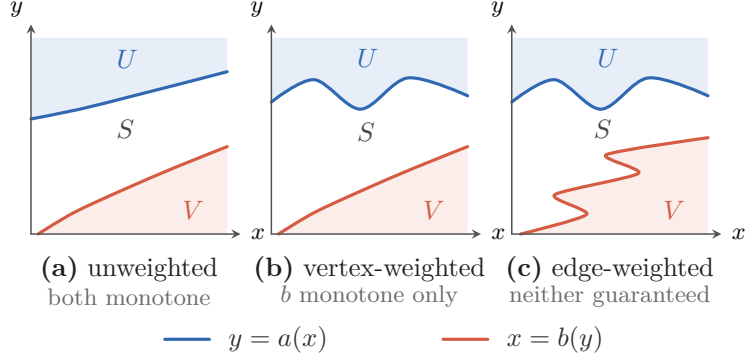
\begin{figure}[tbp]
\centering
\begin{tikzpicture}[
  x=1.18cm,
  y=1.18cm,
  font=\footnotesize,
  line cap=round,
  line join=round,
  axis/.style={
    ->,
    >=stealth,
    draw=black!68,
    line width=.56pt
  },
  upper curve/.style={
    draw=thresholdblue,
    line width=1.20pt
  },
  lower curve/.style={
    draw=thresholdred,
    line width=1.20pt,
  },
  region label/.style={
    font=\small
  },
  panel title/.style={
    font=\small,
    align=center,
    text=black!88
  },
  panel note/.style={
    font=\footnotesize,
    align=center,
    text=black!58
  },
  legend label/.style={
    font=\small,
    anchor=west,
    text=black!82
  }
]

\begin{scope}
  \path[fill=thresholdblue!11]
    plot[smooth,tension=.55] coordinates {
      (0,1.29) (.55,1.40) (1.15,1.55) (1.70,1.69) (2.20,1.82)
    }
    -- (2.20,2.20) -- (0,2.20) -- cycle;
  \path[fill=thresholdred!11]
    plot[smooth,tension=.55] coordinates {
      (.08,0) (.48,.24) (1.02,.49) (1.60,.74) (2.20,.98)
    }
    -- (2.20,0) -- cycle;

  \draw[upper curve]
    plot[smooth,tension=.55] coordinates {
      (0,1.29) (.55,1.40) (1.15,1.55) (1.70,1.69) (2.20,1.82)
    };
  \draw[lower curve]
    plot[smooth,tension=.55] coordinates {
      (.08,0) (.48,.24) (1.02,.49) (1.60,.74) (2.20,.98)
    };

  \draw[axis] (0,0) -- (2.35,0) node[right] {$x$};
  \draw[axis] (0,0) -- (0,2.35) node[above left=-1pt] {$y$};

  \node[region label,text=thresholdblue] at (1.08,1.96) {$U$};
  \node[region label,text=black!72] at (1.06,1.18) {$S$};
  \node[region label,text=thresholdred] at (1.82,.27) {$V$};

  \node[panel title] at (1.10,-.40) {\textbf{(a)} unweighted};
  \node[panel note]  at (1.10,-.70) {both monotone};
\end{scope}

\begin{scope}[xshift=3.18cm]
  \path[fill=thresholdblue!11]
    plot[smooth,tension=.55] coordinates {
      (0,1.48) (.48,1.73) (1.00,1.40) (1.53,1.75) (2.20,1.55)
    }
    -- (2.20,2.20) -- (0,2.20) -- cycle;
  \path[fill=thresholdred!11]
    plot[smooth,tension=.55] coordinates {
      (.08,0) (.48,.24) (1.02,.49) (1.60,.74) (2.20,.98)
    }
    -- (2.20,0) -- cycle;

  \draw[upper curve]
    plot[smooth,tension=.55] coordinates {
      (0,1.48) (.48,1.73) (1.00,1.40) (1.53,1.75) (2.20,1.55)
    };
  \draw[lower curve]
    plot[smooth,tension=.55] coordinates {
      (.08,0) (.48,.24) (1.02,.49) (1.60,.74) (2.20,.98)
    };

  \draw[axis] (0,0) -- (2.35,0) node[right] {$x$};
  \draw[axis] (0,0) -- (0,2.35) node[above left=-1pt] {$y$};

  \node[region label,text=thresholdblue] at (1.08,1.98) {$U$};
  \node[region label,text=black!72] at (1.06,1.15) {$S$};
  \node[region label,text=thresholdred] at (1.82,.27) {$V$};

  \node[panel title] at (1.10,-.40) {\textbf{(b)} vertex-weighted};
  \node[panel note]  at (1.10,-.70) {$b$ monotone only};
\end{scope}

\begin{scope}[xshift=6.36cm]
  \path[fill=thresholdblue!11]
    plot[smooth,tension=.55] coordinates {
      (0,1.48) (.48,1.73) (1.00,1.40) (1.53,1.75) (2.20,1.55)
    }
    -- (2.20,2.20) -- (0,2.20) -- cycle;
  \path[fill=thresholdred!11]
    plot[smooth,tension=.42] coordinates {
      (.10,0) (.84,.22) (.48,.44) (1.42,.68) (1.06,.89) (2.20,1.08)
    }
    -- (2.20,0) -- cycle;

  \draw[upper curve]
    plot[smooth,tension=.55] coordinates {
      (0,1.48) (.48,1.73) (1.00,1.40) (1.53,1.75) (2.20,1.55)
    };
  \draw[lower curve]
    plot[smooth,tension=.42] coordinates {
      (.10,0) (.84,.22) (.48,.44) (1.42,.68) (1.06,.89) (2.20,1.08)
    };

  \draw[axis] (0,0) -- (2.35,0) node[right] {$x$};
  \draw[axis] (0,0) -- (0,2.35) node[above left=-1pt] {$y$};

  \node[region label,text=thresholdblue] at (1.08,1.98) {$U$};
  \node[region label,text=black!72] at (1.02,1.18) {$S$};
  \node[region label,text=thresholdred] at (1.82,.27) {$V$};

  \node[panel title] at (1.10,-.40) {\textbf{(c)} edge-weighted};
  \node[panel note]  at (1.10,-.70) {neither guaranteed};
\end{scope}

\draw[upper curve] (1.50,-1.18) -- (2.00,-1.18);
\node[legend label] at (2.10,-1.18) {$y=a(x)$};
\draw[lower curve] (4.62,-1.18) -- (5.12,-1.18);
\node[legend label] at (5.22,-1.18) {$x=b(y)$};

\end{tikzpicture}
\caption{Threshold geometry.  The blue region \(U\) is where \(u\) is
strictly earlier, the red region \(V\) is where \(v\) is strictly
earlier, and \(S\) is where \(uv\) is selected.  Vertical sections of
\(U\) and horizontal sections of \(V\) are intervals.  Whole-curve
monotonicity holds for both boundaries in unweighted
\textsc{Ranking}, only for \(b\) in the vertex-weighted random-arrival
model, and is not guaranteed in the edge-weighted model.  The curves
are schematic.  Example~\ref{ex:both-cutoffs-nonmonotone} gives a
concrete edge-weighted instance in which both boundaries are
nonmonotone.}
\label{fig:threshold-monotonicity}
\end{figure}

\subsection{Cutoffs}
Fix an instance \(G=(L,R,w,E)\), a target edge \(uv\in E\), and the
ranks of all vertices other than the ones explicitly varied.  For a
vertex \(z\), write \(G-z\) for the induced instance obtained by
deleting \(z\) and restricting \(w\) and \(E\).  We compare vertices
by the order in which they become matched in the resulting total scan,
treating an unmatched vertex as later than every matched vertex.
Earlier scan positions have weakly larger priority.  Because the fixed
target edge \(uv\) is realized, at least one of \(u\) and \(v\) is
eventually matched, and they become matched at the same step if and
only if \(uv\) is selected.  Hence the event that \(v\) is matched no
later than \(u\) is the disjoint union of the events that \(v\) is
strictly earlier than \(u\) and that \(uv\) is selected.

Changing one rank can reorder the incident edges of that vertex, so the
usual one-line monotonicity proof for one-sided perturbed greedy does
not apply.  The following weaker statement is exactly what the
two-curve analysis needs.

\begin{lemma}\label{lem:one-sided}
Fix all ranks except \(x_v\).  If \(u\) is matched strictly earlier than \(v\)
when \(x_v=t\), then \(u\) is matched strictly earlier than \(v\) for every
\(x_v=t'>t\).
\end{lemma}

\begin{proof}
Consider the execution with \(x_v=t\), and suppose that \(u\) is matched
through the edge \(uz\). At the time when \(uz\) is probed, every probed pair
incident to \(v\) is rejected. Consequently, the decisions made up to and
including the probe of \(uz\) are unaffected by the presence of \(v\);
equivalently, running the greedy algorithm on \(G-v\) produces the same
matching up to the probe of \(uz\), and in particular selects \(uz\).

Now increase the rank of \(v\) from \(t\) to \(t'>t\). By the tie-breaking property, every pair incident to \(v\) moves weakly later 
relative to every pair not incident to \(v\), although the relative order
among the \(v\)-incident pairs may change. Therefore, no pair that was
probed before \(uz\) in the original execution is displaced by a
\(v\)-incident pair. Since the execution restricted to \(G-v\) is
unchanged, the edge \(uz\) is still probed while both \(u\) and \(z\) are
unmatched, and hence \(u\) is matched at the same time as before.
Because all pairs incident to \(v\) can only move later, \(v\) cannot be
matched before the probe of \(uz\). Thus \(u\) remains matched strictly
earlier than \(v\), as claimed.
\end{proof}

\begin{corollary}\label{cor:cutoff}
For every fixed \(x=x_u\), there is a cutoff \(a(x)\in[0,1]\) such
that, away from the boundary \(y=a(x)\),
\[
 \begin{cases}
  y>a(x) &\Longrightarrow u\text{ is strictly earlier than }v,\\
  y<a(x) &\Longrightarrow v\text{ is no later than }u.
 \end{cases}
\]
Symmetrically, for every fixed \(y=x_v\), there is \(b(y)\in[0,1]\)
such that \(x>b(y)\) exactly describes the region where \(v\) is
strictly earlier than \(u\).
\end{corollary}

\begin{proof}
For fixed \(x\), Lemma~\ref{lem:one-sided} makes the set of \(y\) for
which \(u\) is strictly earlier than \(v\) upward closed.  Its infimum,
with the natural conventions for the empty and full sets, is the
required \(a(x)\).  Interchanging \(u\) and \(v\) gives \(b(y)\).
\end{proof}

\subsubsection{Monotonicity of Cutoffs}
We emphasize that no monotonicity assumption is imposed on the cutoff
functions \(a\) and \(b\). Indeed, we provide a concrete example showing
that both cutoff curves can be non-monotone.

This stands in contrast to the vertex-weighted setting studied by \citet{HuangTangWuZhang2019}, where \(b\) is shown to be nondecreasing. By symmetry, in the unweighted setting both cutoff functions are therefore necessarily monotone; see Fig.~\ref{fig:threshold-monotonicity}. We return to this monotonicity property in Section~\ref{sec:online}.

\begin{example}
\label{ex:both-cutoffs-nonmonotone}
Let
\[
 h(t)=1-t,
 \qquad
 L=\{u,q\},
 \qquad
 R=\{v,r\},
\]
with target edge \(uv\).  Take
\[
 E=\{uv,ur,qv\},
 \qquad
 (w_{uv},w_{ur},w_{qv},w_{qr})=(1,2,2,0),
 \qquad
 x_q=x_r=\frac34.
\]
Vary
\[
 x=x_u,
 \qquad
 y=x_v,
 \qquad
 A=1-x,
 \qquad
 B=1-y.
\]
The three relevant priorities are
\[
 P:=\pi_{uv}=\frac{AB}{A+B},
 \qquad
 P_u:=\pi_{ur}=\frac{2A}{4A+1},
 \qquad
 P_v:=\pi_{qv}=\frac{2B}{4B+1}.
\]

\begin{figure}[t]
\centering
\begin{tikzpicture}[
  x=5.8cm,
  y=5.8cm,
  font=\footnotesize,
  line cap=round,
  line join=round,
  axis/.style={
    ->,
    >=stealth,
    draw=black!62,
    line width=.55pt
  },
  unit edge/.style={
    draw=black!16,
    line width=.55pt
  },
  a curve/.style={
    draw=thresholdblue,
    line width=1.15pt
  },
  b curve/.style={
    draw=thresholdred,
    line width=1.15pt,
  },
  shared curve/.style={
    draw=sharedviolet,
    line width=1.35pt
  },
  region label/.style={
    font=\small\bfseries
  },
  point label/.style={
    font=\scriptsize,
    inner sep=1pt
  },
  legend label/.style={
    anchor=west,
    font=\scriptsize,
    text=black!68
  }
]

\path[fill=thresholdblue!8]
  plot[domain=0:.25,samples=70,variable=\t]
    ({\t},{(1-2*\t)/(3-4*\t)})
  -- (1,1) -- (0,1) -- cycle;

\path[fill=thresholdred!8]
  plot[domain=0:.25,samples=70,variable=\t]
    ({(1-2*\t)/(3-4*\t)},{\t})
  -- (1,1) -- (1,0) -- cycle;

\path[fill=black!3]
  (0,0) -- ({1/3},0)
  plot[domain=0:.25,samples=70,variable=\t]
    ({(1-2*\t)/(3-4*\t)},{\t})
  plot[domain=.25:0,samples=70,variable=\t]
    ({\t},{(1-2*\t)/(3-4*\t)})
  -- cycle;

\draw[unit edge] (0,1) -- (1,1) -- (1,0);
\draw[axis] (0,0) -- (1.065,0) node[right] {$x$};
\draw[axis] (0,0) -- (0,1.065) node[above] {$y$};
\node[below left=2pt,text=black!62] at (0,0) {$0$};
\node[below=2pt,text=black!62] at (1,0) {$1$};
\node[left=2pt,text=black!62] at (0,1) {$1$};

\draw[a curve]
  plot[domain=0:.25,samples=70,variable=\t]
    ({\t},{(1-2*\t)/(3-4*\t)});
\draw[b curve]
  plot[domain=0:.25,samples=70,variable=\t]
    ({(1-2*\t)/(3-4*\t)},{\t});

\draw[shared curve] (.25,.25) -- (1,1);

\fill[thresholdblue] (0,{1/3}) circle (1.35pt);
\fill[thresholdred] ({1/3},0) circle (1.35pt);
\fill[sharedviolet] (.25,.25) circle (1.55pt);
\fill[white] (1,1) circle (1.8pt);
\draw[sharedviolet,line width=.6pt] (1,1) circle (1.8pt);

\node[point label,text=thresholdblue,anchor=east]
  at (-.018,{1/3}) {$\left(0,\frac13\right)$};
\node[point label,text=thresholdred,anchor=north]
  at ({1/3},-.022) {$\left(\frac13,0\right)$};
\node[point label,anchor=north east,text=black!70]
  at (.235,.235) {$\left(\frac14,\frac14\right)$};

\node[region label,text=thresholdblue] at (.23,.66) {$U$};
\node[region label,text=thresholdred] at (.66,.23) {$V$};
\node[region label,text=black!68] at (.105,.105) {$S$};

\node[point label,text=sharedviolet,anchor=west]
  (sharednote) at (.50,.84) {$a(t)=b(t)=t$};
\draw[sharedviolet!72,line width=.45pt]
  (sharednote.south west) -- (.64,.64);

\draw[a curve] (-.03,-.165) -- (.07,-.165);
\node[legend label] at (.10,-.165) {$y=a(x)$};
\draw[b curve] (.42,-.165) -- (.52,-.165);
\node[legend label] at (.55,-.165) {$x=b(y)$};
\draw[shared curve] (.86,-.165) -- (.96,-.165);
\node[legend label] at (.99,-.165) {shared};

\end{tikzpicture}
\caption{Exact cutoff curves for
Example~\ref{ex:both-cutoffs-nonmonotone}.  The blue boundary is
\(y=a(x)\), the red boundary is \(x=b(y)\), and the bicolored diagonal
is shared: \(a(t)=b(t)=t\) for \(1/4\le t<1\).  Each cutoff decreases
from \(1/3\) to \(1/4\) before increasing along the diagonal.  The
shading shows the interiors of \(U\) and \(V\); boundary membership is
tie-dependent.  The open circle records that ranks lie in \([0,1)\).}
\label{fig:example-nonmonotone-curves}
\end{figure}
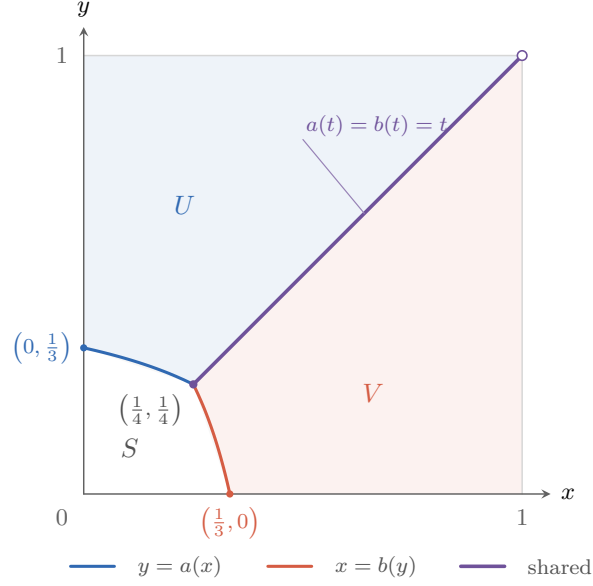

Away from priority ties, the region $U$ occurs when $P_u>\max\{P,P_v\}$,
the region $V$ occurs when $P_v>\max\{P,P_u\}$, and $uv$ is selected when $P>\max\{P_u,P_v\}$. Boundary membership depends on the tie-breaking rule and does not
affect the cutoff values below.
A direct
comparison gives, up to boundary conventions,
\[
 a(t)=b(t)=
 \begin{cases}
  \dfrac{1-2t}{3-4t},&0\le t\le \dfrac14,\\[6pt]
  t,&\dfrac14\le t<1.
 \end{cases}
\]
In particular,
\[
 a(0)=b(0)=\frac13,
 \qquad
 a\!\left(\frac14\right)
 =b\!\left(\frac14\right)
 =\frac14.
\]
Both curves decrease on \([0,1/4]\) and increase afterward, so they
are simultaneously nonmonotone in this instance.
Figure~\ref{fig:example-nonmonotone-curves} plots their exact geometry.


This example is size-minimal in the natural sense.  A nonempty
\(U\)-region requires an alternative realized edge incident to \(u\),
and a nonempty \(V\)-region requires one incident to \(v\).  A simple
bipartite graph therefore needs an additional vertex on each side and
at least the three realized edges \(uv\), \(ur\), and \(qv\).
\end{example}

\subsection{Alternating paths}

The cutoff alone compares the order in which the two endpoints are
matched.  To lower-bound the gain of a specific endpoint, we also need
to know what happens when the opposite vertex and all of its incident
edges are restored.

\begin{lemma}\label{lem:insertion}
Let \(v\in R\), and fix a total order on the edges of a bipartite graph
\(G\).  Compare greedy on \(G\) with greedy on \(G-v\), where the latter
skips the \(v\)-incident edges in the same order.  Whenever a vertex
\(z\in L\) is matched in the reduced run, it has already been matched
in the full run by the point when the reduced run selects \(z\)'s
matching edge.
\end{lemma}

\begin{proof}
Couple the two scans in their common edge order, skipping
\(v\)-incident edges in the reduced run.  The first discrepant
acceptance, if any, must be an edge \(vz_0\) accepted only in the full
run.  Thereafter an edge off the resulting alternating path sees the
same availability in both runs, so each discrepant acceptance must
extend the path at its exposed endpoint.  The partial matchings prevent
branching.  Moreover, every interior vertex is matched in both runs
and \(v\) is unavailable in the full run, so no accepted edge can
close a cycle.  Induction therefore shows that the discrepant accepted
edges occur in the order
\[
 v-z_0-y_1-z_1-y_2-z_2-\cdots
\]
and alternate between the two runs:
\[
 vz_0
 \succ z_0y_1
 \succ z_1y_1
 \succ z_1y_2
 \succ z_2y_2
 \succ\cdots,
\]
where \(e\succ f\) means that \(e\) is probed first.  Thus each
\(L\)-vertex on the path is matched in the full run before the edge
that matches it in the reduced run is probed.  Every \(L\)-vertex off
the path has the same status in both runs.  If no \(v\)-incident edge
is accepted, the runs coincide.  The claim follows.
\end{proof}

\begin{remark}
The symmetric difference between two arbitrary rank realizations need
not be a path; it may be an alternating even cycle.
Lemma~\ref{lem:insertion} uses the more controlled comparison between a graph
and the graph obtained by deleting one vertex.  That is the comparison
needed below.
\end{remark}

\subsection{An edgewise threshold bound}

Recall the sampled-domain gain-sharing rule \(g_h\) from~\eqref{eq:g-def}.  When a
cutoff takes the unsampled boundary value \(1\), we use the analytical
extension
\begin{equation}\label{eq:g-boundary-extension}
 g_h(t,1):=0,
 \qquad
 g_h(1,t):=1\quad(t<1),
 \qquad
 g_h(1,1):=\frac12.
\end{equation}
These values are never inputs to the algorithm.  Fix \(x=x_u<1\), let
\(a=a(x)\), and run the reference process on \(G-v\).

\begin{lemma}\label{lem:threshold-gain}
If \(uv\) is not selected, then
\begin{equation}\label{eq:threshold-u}
 \alpha_u\ge g_h(x,a(x))w_{uv}.
\end{equation}
If \(uv\) is selected, then \(\alpha_u+\alpha_v=w_{uv}\).
The symmetric statement holds for \(v\) and \(b(y)\).
\end{lemma}

\begin{proof}
If \(a=1\), the right-hand side of~\eqref{eq:threshold-u} is zero, so
assume \(a<1\).

Fix \(y>a\).  Then \(u\) is matched strictly before \(v\) through a
non-\(v\) edge.  Up to its selection, every probed \(v\)-incident edge
has changed no state, so the same edge matches \(u\) in the reference
run on \(G-v\).  The reference edge, its numerical priority \(P\), and
the relative order of all non-\(v\) pairs are independent of \(y\) under the tie-breaking
convention.  Necessarily
\(P\ge\pi_{uv}(x,y)\) for every \(y>a\);
otherwise \(uv\) would be probed while both endpoints were free.
Right-continuity as \(y\downarrow a\) gives
\[
 P\ge\pi_{uv}(x,a)
  =h(x)g_h(x,a)w_{uv}.
\]

By Lemma~\ref{lem:insertion}, when the common scan reaches \(u\)'s
reference edge, \(u\) is already matched in \(G\).  Unless \(uv\) is
selected, the priority of the edge matching \(u\) is at least \(P\), so
\[
 \alpha_u\ge\frac{P}{h(x)}
 \ge g_h(x,a)w_{uv}.
\]
The statement for \(v\) is symmetric.
\end{proof}

\section{The two-curve variational game}\label{sec:game}

Continue conditioning on all ranks except \(x=x_u\) and \(y=x_v\).
Define
\begin{align}
 U(x,y)&:=\one\{y>a(x)\},\label{eq:U}\\
 \widetilde V(x,y)&:=\one\{x>b(y)\},\label{eq:V}\\
 S(x,y)&:=(1-U(x,y))(1-\widetilde V(x,y)).\label{eq:S}
\end{align}
The first two events mean, respectively, that \(u\) is strictly earlier
and that \(v\) is strictly earlier.  They are disjoint.  Since \(uv\)
is a realized edge, if neither endpoint is earlier then \(uv\) is
selected.  Consequently, almost everywhere,
\begin{equation}\label{eq:partition}
 U(x,y)\widetilde V(x,y)=0,
 \qquad
 U(x,y)+\widetilde V(x,y)+S(x,y)=1.
\end{equation}

Lemma~\ref{lem:threshold-gain} gives
\begin{align}
 \frac{\alpha_u}{w_{uv}}
 &\ge g_h(x,a(x))\bigl(U(x,y)+\widetilde V(x,y)\bigr)
       +g_h(x,y)S(x,y),\label{eq:alpha-u-pointwise}\\
 \frac{\alpha_v}{w_{uv}}
 &\ge g_h(y,b(y))\bigl(U(x,y)+\widetilde V(x,y)\bigr)
       +g_h(y,x)S(x,y).\label{eq:alpha-v-pointwise}
\end{align}
Consequently,
\begin{equation}\label{eq:pointwise}
 \frac{\alpha_u+\alpha_v}{w_{uv}}
 \ge
 S(x,y)+(U(x,y)+\widetilde V(x,y))(g_h(x,a(x))+g_h(y,b(y))).
\end{equation}

\begin{definition}\label{def:threshold-pairs}
Let \(\mathcal T\) be the measurable pairs \(a,b:[0,1]\to[0,1]\)
whose indicators~\eqref{eq:U}--\eqref{eq:V} satisfy
\(U(x,y)\widetilde V(x,y)=0\) almost everywhere. 
\end{definition}

For every gain-sharing function $g$ and pair $(a,b) \in \mathcal{T}$, define
\begin{equation}\label{eq:Gamma}
 \Gamma(g,a,b)
 :=
 \int_0^1\!\!\int_0^1
 \left[
  S(x,y)+\bigl(U(x,y)+\widetilde V(x,y)\bigr)
  \bigl(g(x,a(x))+g(y,b(y))\bigr)
 \right]\dd x\dd y.
\end{equation}
For a score function \(h\), define
\begin{equation}\label{eq:fixed-score-value}
 \gamma_{\mathrm{EW}}(h)
 :=
 \inf_{(a,b)\in\mathcal T}\Gamma(g_h,a,b).
\end{equation}

\begin{proposition}
\label{prop:fixed-score-guarantee}
For every score function \(h\) and every edge-weighted
oblivious instance \(G\),
\[
 \E[\ALG(G)]\ge\gamma_{\mathrm{EW}}(h)\,\OPT(G).
\]
\end{proposition}

\begin{proof}
Fix a positive-weight realized edge \(uv\), and condition on all ranks
other than \(x=x_u\) and \(y=x_v\).  The resulting cutoff pair
\((a,b)\) belongs to \(\mathcal T\) by
Definition~\ref{def:threshold-pairs}.  Integrating~\eqref{eq:pointwise} gives
\[
 \frac{
  \E_{x,y}[\alpha_u+\alpha_v
   \mid\boldsymbol x_{-\{u,v\}}]
 }{w_{uv}}
 \ge
 \Gamma(g_h,a,b)
 \ge
 \inf_{(a,b)\in\mathcal T}\Gamma(g_h,a,b)
 =
 \gamma_{\mathrm{EW}}(h).
\]
After averaging over the conditioned ranks, sum this bound over a
maximum-weight matching \(M^\star\):
\[
 \gamma_{\mathrm{EW}}(h)\OPT(G)
 \le
 \sum_{uv\in M^\star}\E[\alpha_u+\alpha_v]
 \le
 \E\!\left[\sum_z\alpha_z\right]
 =
 \E[\ALG(G)].
\]
Zero-weight edges require no separate argument.
\end{proof}

The factor-revealing value of the score-balanced family is therefore
\begin{equation}\label{eq:maxmin}
 \gamma_{\mathrm{EW}}
 :=
 \sup_h\gamma_{\mathrm{EW}}(h)
 =
 \sup_h\inf_{(a,b)\in\mathcal T}\Gamma(g_h,a,b),
\end{equation}
where \(g_h\) is defined on sampled ranks in~\eqref{eq:g-def} and has
the boundary extension~\eqref{eq:g-boundary-extension}.
Equation~\eqref{eq:maxmin} is a max--min problem over one score function
and two arbitrary threshold curves.  The inner problem is infinite
dimensional.  Actual realizable cutoff pairs form a subset of
\(\mathcal T\); the infimum over \(\mathcal T\) is therefore a
relaxation and gives a certified lower bound.  We do not claim that
every pair in \(\mathcal T\) is realizable by an instance.  The
supremum ranges over the admissible nonincreasing, right-continuous
score functions in Section~\ref{sec:model}, including step functions.
From Section~\ref{sec:finite} onward we restrict to score functions
that are constant on a uniform grid.  The resulting cut and flow
problems are finite and satisfy strong duality.

\section{The finite-grid variational game and exact cut--flow duality}
\label{sec:finite}

The two-curve game in Section~\ref{sec:game} is stated on the full rank
square.  From this point onward we work at a fixed grid resolution.
Both the score and the cut indicators are represented cell by cell, so
every optimization problem below is finite-dimensional.
Proposition~\ref{prop:step-flow-bridge} shows that, for every grid size
and every height vector satisfying~\eqref{eq:finite-scores}, the
maximum-flow value lower-bounds \(\gamma_{\mathrm{EW}}(h^{(m)})\), and
hence \(\gamma_{\mathrm{EW}}\).  Together with the fixed-score
variational guarantee, this converts every finite flow certificate
directly into an algorithmic guarantee.

The indicator lifting turns the exponentially many threshold pairs
into the cuts of one polynomial-size directed network.  The associated
flow problem is the ordinary maximum-flow problem on that same
network.  Finite max-flow/min-cut duality therefore gives equality and
attainment on both sides.  We first derive the network for an arbitrary
step score; Section~\ref{sec:certificate} then supplies an explicit
\(240\)-step score and a reproducible exact-integer verification.

\subsection{Step scores and discrete integration by parts}

The network below represents a finite factor-revealing relaxation for
a step score.  Fix \(m\ge1\), write
\([m]:=\{1,\ldots,m\}\), and partition
\([0,1)\) into cells
\[
 J_i:=\left[\frac{i-1}{m},\frac{i}{m}\right)
 \qquad(i\in[m]).
\]
Choose step heights
\begin{equation}\label{eq:finite-scores}
 h_1\ge h_2\ge\cdots\ge h_m>0,
 \qquad h_{m+1}:=0,
\end{equation}
and define
\[
 h^{(m)}(x):=h_i\quad(x\in J_i),
 \qquad h^{(m)}(1):=0.
\]
The value \(h_{m+1}=0\) is an unsampled boundary convention used only
for bookkeeping.

For this choice of \(h\), the step-score guarantee is certified by an
ordinary finite maximum-flow/minimum-cut calculation.  No limiting
argument or numerical quadrature is needed.

For \(i,j\in[m]\), define
\begin{equation}\label{eq:finite-g}
 g_{ij}:=\frac{h_j}{h_i+h_j},
 \qquad
 g_{i,m+1}:=0,
\end{equation}
where the last equality is the discrete counterpart of the zero
endpoint condition \(g(x,1)=0\).
Thus right-continuity gives
\[
 g_{h^{(m)}}\!\left(x,\frac{a}{m}\right)
 =g_{i,a+1}
 \qquad
 (x\in J_i,\ a\in\{0,\ldots,m\}).
\]
Also define
\begin{align}
 q_{ij}
 &:=g_{ij}-g_{i,j+1},\label{eq:finite-q}\\
 r_{ij}
 &:=(m-j+1)g_{ij}-(m-j)g_{i,j+1}.\label{eq:finite-r}
\end{align}
These are exact finite increments: \(q_{ij}\) records the drop of \(g\) across a cell boundary, and \(r_{ij}\) is the corresponding
discrete integration-by-parts coefficient.
The score order makes \(q_{ij}\ge0\), and
\begin{equation}\label{eq:finite-q-linear}
 g_{ij}=\sum_{k=j}^m q_{ik},
 \qquad
 r_{ij}
 =g_{ij}+(m-j)q_{ij}
 =(m-j+1)q_{ij}+\sum_{k=j+1}^m q_{ik}.
\end{equation}
Moreover, discrete complementarity is the linear identity
\[
 \sum_{k=j}^m q_{ik}+\sum_{k=i}^m q_{jk}=1.
\]
Hence \(r_{ij}\ge0\); every \(q\)-dependent objective-arc capacity is
linear in \(\boldsymbol q=(q_{ij})\), and all capacities are affine in
\(\boldsymbol q\).
In particular, the unsampled bookkeeping height supplies the final
increment
\[
 q_{i,m}=r_{i,m}=g_{i,m}.
\]

Write \(\mathcal Q_m:=\{0,\ldots,m\}^m\).  For
\((\boldsymbol a,\boldsymbol b)\in\mathcal Q_m^2\), define
\[
 U_{ij}:=\one\{j>a_i\},
 \qquad
 \widetilde V_{ij}:=\one\{i>b_j\}.
\]
The admissible cellwise threshold pairs are
\begin{equation}\label{eq:finite-domain}
 \mathcal T_m:=
 \left\{
  (\boldsymbol a,\boldsymbol b)\in\mathcal Q_m^2:
  U_{ij}\widetilde V_{ij}=0\ \text{for all }i,j
 \right\}.
\end{equation}
For \((\boldsymbol a,\boldsymbol b)\in\mathcal T_m\), put
\begin{equation}\label{eq:finite-ABM}
 P_i:=g_{i,a_i+1},
 \qquad
 Q_j:=g_{j,b_j+1},
 \qquad
 S_{ij}:=(1-U_{ij})(1-\widetilde V_{ij}).
\end{equation}
The telescoping identities
\begin{align}
 P_i&=\sum_{j=1}^m q_{ij}U_{ij},\label{eq:finite-A}\\
 (m-a_i)P_i&=\sum_{j=1}^m r_{ij}U_{ij}\label{eq:finite-aA}
\end{align}
follow immediately from~\eqref{eq:finite-q}--\eqref{eq:finite-r};
the analogous identities hold for \(Q,\widetilde V\).

\subsection{The exact cell network}

Define binary labels
\begin{equation}\label{eq:finite-XY}
 X_{ij}:=U_{ij},
 \qquad
 Y_{ij}:=1-\widetilde V_{ij}.
\end{equation}
Every threshold pair satisfies
\begin{align}
 X_{ij}&\le Y_{ij},\label{eq:finite-local}\\
 X_{ij}&\le X_{i,j+1} &&(j<m),\label{eq:finite-X}\\
 Y_{i+1,j}&\le Y_{ij} &&(i<m).\label{eq:finite-Y}
\end{align}
Conversely, every binary pair satisfying
\eqref{eq:finite-local}--\eqref{eq:finite-Y} is induced by some
threshold vectors \(a,b\).  Notice again that the vectors
\(i\mapsto a_i\) and \(j\mapsto b_j\) are not required to be monotone.
Source-side membership will encode label one.  The local implication
leaves exactly three states:
\[
\begin{array}{@{}ccl@{}}
\toprule
(X_{ij},Y_{ij})&(U_{ij},S_{ij},\widetilde V_{ij})&\text{region}\\
\midrule
(1,1)&(1,0,0)&U\\
(0,1)&(0,1,0)&S\\
(0,0)&(0,0,1)&V\\
\bottomrule
\end{array}
\]
The last two implications link these states into a chain within each
row and column, but do not link distinct row cutoffs or distinct
column cutoffs.  Figure~\ref{fig:finite-threshold-network} previews one
representative of every arc family before we list their capacities.

\begin{figure}[tbp]
\centering
\begin{tikzpicture}[
  font=\footnotesize,
  cellnode/.style={circle,draw=black!45,fill=white,minimum size=9mm,
                   inner sep=.6pt},
  focus/.style={cellnode,draw=thresholdblue!75,fill=thresholdblue!5,
                line width=.8pt},
  terminal/.style={circle,draw=black!55,fill=black!4,
                   minimum size=8mm},
  objective/.style={->,>=stealth,draw=thresholdblue,
                    line width=1.05pt,shorten >=1pt,shorten <=1pt},
  hard/.style={->,>=stealth,draw=thresholdred,line width=1.05pt,
               dash pattern=on 4pt off 2.5pt,
               shorten >=1pt,shorten <=1pt},
  lab/.style={fill=white,inner sep=1.5pt,font=\scriptsize},
  panel/.style={draw=black!22,rounded corners=2pt,fill=black!1}
]
  \draw[panel] (-.2,-.4) rectangle (4.25,4.0);
  \draw[panel] (7.0,-.4) rectangle (11.45,4.0);
  \node[font=\small\bfseries] at (2.0,4.35) {\(\mathsf Y\)-grid};
  \node[font=\small\bfseries] at (9.2,4.35) {\(\mathsf X\)-grid};

  \node[terminal] (s) at (-1.05,1.45) {\(s\)};
  \node[focus] (Yij) at (1.6,1.45) {\(\mathsf Y_{ij}\)};
  \node[cellnode] (Yik) at (1.6,3.15) {\(\mathsf Y_{ik}\)};
  \node[cellnode] (Ykj) at (3.35,1.45) {\(\mathsf Y_{kj}\)};
  \node[cellnode] (Yipj) at (.45,.30) {\(\mathsf Y_{i+1,j}\)};

  \node[focus] (Xij) at (8.6,1.45) {\(\mathsf X_{ij}\)};
  \node[cellnode] (Xijp) at (8.6,3.15) {\(\mathsf X_{i,j+1}\)};
  \node[terminal] (t) at (12.25,1.45) {\(t\)};

  \draw[objective] (s) -- node[lab,above] {\(r_{ji}/m^2\)} (Yij);
  \draw[objective] (Yij) to[bend left=8]
       node[lab,above] {\(1/m^2\)} (Xij);
  \draw[objective] (Xij) -- node[lab,above] {\(r_{ij}/m^2\)} (t);
  \draw[objective] (Xij) to[bend right=20]
       node[lab,above,sloped] {\(q_{ij}/m^2\)} (Yik);
  \draw[objective] (Xij) to[bend left=10]
       node[lab,above] {\(q_{jk}/m^2\)} (Ykj);

  \draw[hard] (Xij) to[bend left=18]
       node[lab,below] {\(\infty\)} (Yij);
  \draw[hard] (Xij) -- node[lab,right] {\(\infty\)} (Xijp);
  \draw[hard] (Yipj) -- node[lab,below,sloped] {\(\infty\)} (Yij);

  \draw[objective] (3.25,-.9) -- (4.05,-.9);
  \node[anchor=west,text=black!65] at (4.2,-.9) {objective arc};
  \draw[hard] (7.35,-.9) -- (8.15,-.9);
  \node[anchor=west,text=black!65] at (8.3,-.9) {hard implication};
\end{tikzpicture}
\caption{Representative arcs of the finite threshold network.  The
two boxes are independent copies of the cell grid \([m]^2\).  The two
blue arcs from \(\mathsf X_{ij}\) into the \(\mathsf Y\)-grid stand
for the full \(k\)-indexed fans in~\eqref{eq:finite-network}; the red
arcs enforce the three implications in~\eqref{eq:finite-hard}.}
\label{fig:finite-threshold-network}
\end{figure}
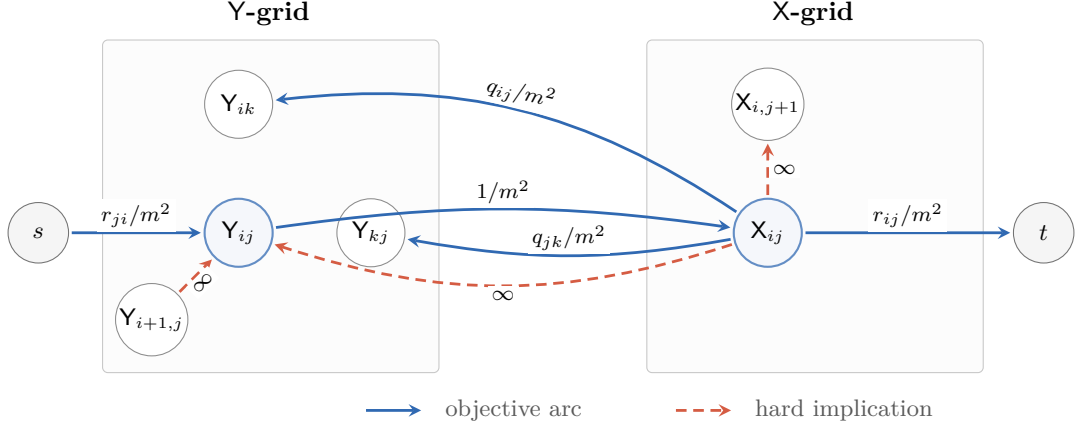

For \((\boldsymbol a,\boldsymbol b)\in\mathcal T_m\), abbreviate
\begin{equation}\label{eq:finite-Gamma}
 \Gamma_m^{\boldsymbol h}(\boldsymbol a,\boldsymbol b)
 :=\frac1{m^2}\sum_{i,j=1}^m
 \left[
  S_{ij}+(U_{ij}+\widetilde V_{ij})(P_i+Q_j)
 \right].
\end{equation}

Create \(2m^2\) nodes
\(\mathsf X_{ij},\mathsf Y_{ij}\), together with source \(s\) and sink
\(t\).  After normalization by the cell area \(m^{-2}\), the objective
arcs are
\begin{equation}\label{eq:finite-network}
\begin{array}{@{}lll@{}}
\toprule
\text{arc} & \text{capacity} & \text{cut contribution}\\
\midrule
s\to\mathsf Y_{ij}
 &r_{ji}/m^2&r_{ji}\widetilde V_{ij}/m^2\\
\mathsf Y_{ij}\to\mathsf X_{ij}
 &1/m^2&S_{ij}/m^2\\
\mathsf X_{ij}\to t
 &r_{ij}/m^2&r_{ij}U_{ij}/m^2\\
\mathsf X_{ij}\to\mathsf Y_{ik}&q_{ij}/m^2
   &q_{ij}U_{ij}\widetilde V_{ik}/m^2\quad(k\in[m])\\
\mathsf X_{ij}\to\mathsf Y_{kj}&q_{jk}/m^2
   &q_{jk}U_{ij}\widetilde V_{kj}/m^2\quad(k\in[m])\\
\bottomrule
\end{array}
\end{equation}
The hard implication arcs are
\begin{equation}\label{eq:finite-hard}
 \mathsf X_{ij}\to\mathsf Y_{ij},\qquad
 \mathsf X_{ij}\to\mathsf X_{i,j+1},\qquad
 \mathsf Y_{i+1,j}\to\mathsf Y_{ij},
\end{equation}
whenever the indices exist, each with capacity \(+\infty\).  Adjacent
row and column arcs suffice, because their transitive closure enforces
all within-section threshold precedences.  We add no cross-coordinate
arcs encoding monotonicity of \(i\mapsto a_i\) or \(j\mapsto b_j\):
neither property is available in the edge-weighted setting.  Thus the
three hard-arc families encode, respectively, exclusivity, the row
cutoff structure, and the column cutoff structure.

A cut of value at most one
always exists, so replacing every \(+\infty\) in
\eqref{eq:finite-hard} by capacity \(2\) leaves the minimum cut
unchanged.  Let \(\mathcal N_m(\boldsymbol h)\) denote this equivalent
ordinary finite-capacity network, and let
\(\MC_m(\boldsymbol h)\) and \(\MF_m(\boldsymbol h)\) denote its
minimum-cut and maximum-flow values.

\begin{proposition}
\label{prop:finite-cut}
For every fixed \(m\)-level step profile,
\[
 \MF_m(\boldsymbol h)=\MC_m(\boldsymbol h)
 =\min_{(\boldsymbol a,\boldsymbol b)\in\mathcal T_m}
   \Gamma_m^{\boldsymbol h}(\boldsymbol a,\boldsymbol b).
\]
\end{proposition}

\begin{proof}
The network has \(2m^2+2\) nodes and exactly
\(2m^3+6m^2-2m\) arcs, including both objective and implication
families and counting parallel and zero-capacity arcs separately.
For each \(i\), telescoping gives
\[
 \sum_{j=1}^m r_{ij}=m g_{i1}.
\]
Hence both the source-only cut and the cut with every nonterminal on
the source side have capacity at most one.  A minimum cut crosses no
capacity-\(2\) implication arc, so its labels satisfy
\eqref{eq:finite-local}--\eqref{eq:finite-Y} and correspond to an
admissible threshold pair.  For any such cut, its capacity \(C(X,Y)\)
is
\begin{equation}\label{eq:finite-cut-expanded}
\begin{aligned}
 m^2 C(X,Y)
 =\sum_{i,j=1}^m\biggl[
  &Y_{ij}(1-X_{ij})
   +r_{ij}X_{ij}
   +r_{ji}(1-Y_{ij})\\
  &+\sum_{k=1}^m q_{ij}X_{ij}(1-Y_{ik})
   +\sum_{k=1}^m q_{jk}X_{ij}(1-Y_{kj})
 \biggr].
\end{aligned}
\end{equation}
Substitute \(X_{ij}=U_{ij}\) and
\(Y_{ij}=1-\widetilde V_{ij}\).  The first term is \(S_{ij}\),
\eqref{eq:finite-aA} converts the two \(r\)-terms into
\(\sum_{ij}U_{ij}P_i\) and \(\sum_{ij}\widetilde V_{ij}Q_j\); and
\eqref{eq:finite-A} converts the two interaction families into
\(\sum_{ij}\widetilde V_{ij}P_i\) and \(\sum_{ij}U_{ij}Q_j\).
Thus~\eqref{eq:finite-cut-expanded} equals
\(m^2\Gamma_m^{\boldsymbol h}(\boldsymbol a,\boldsymbol b)\).
Conversely, every admissible threshold pair defines a cut with this
value.  Minimizing over cuts gives the second equality, and ordinary
max-flow/min-cut duality gives the first.
\end{proof}

\begin{proposition}
\label{prop:step-flow-bridge}
For every \(m\ge1\) and every \(\boldsymbol h\) satisfying
\eqref{eq:finite-scores},
\[
 \gamma_{\mathrm{EW}}
 \ge \gamma_{\mathrm{EW}}(h^{(m)})
 \ge \MF_m(\boldsymbol h).
\]
\end{proposition}

\begin{proof}
The first inequality follows from~\eqref{eq:maxmin}.  For the second,
fix any \((a,b)\in\mathcal T\).
Independently choose one representative
\(\xi_i\sim\Unif(J_i)\) and one representative
\(\zeta_j\sim\Unif(J_j)\) for each cell.  The representatives are
strictly ordered by their cell indices, and the finitely many sampled
pairs almost surely avoid the null overlap set.  Hence the arrays
\(U(\xi_i,\zeta_j)\) and \(\widetilde V(\xi_i,\zeta_j)\) form
nonoverlapping row and column thresholds and define some
\((\widehat{\boldsymbol a},\widehat{\boldsymbol b})\in\mathcal T_m\).

The sampled boundary \(\widehat a_i/m\) is an endpoint of the cell
containing \(a(\xi_i)\), with the endpoint \(1\) left unchanged; the
same holds for \(\widehat b_j/m\).
Cellwise constancy and monotonicity of \(g_{h^{(m)}}\) in its second coordinate
therefore give
\[
 g_{i,\widehat a_i+1}
 \le g_{h^{(m)}}(\xi_i,a(\xi_i)),
 \qquad
 g_{j,\widehat b_j+1}
 \le g_{h^{(m)}}(\zeta_j,b(\zeta_j)).
\]
The sampled \(U,\widetilde V,S\) indicators agree with the discrete
indicators induced by
\((\widehat{\boldsymbol a},\widehat{\boldsymbol b})\).
Therefore
\[
\begin{aligned}
 \MF_m(\boldsymbol h)
 &=\MC_m(\boldsymbol h)
 \le
 \Gamma_m^{\boldsymbol h}
   (\widehat{\boldsymbol a},\widehat{\boldsymbol b})\\
 &\le
 \frac1{m^2}\sum_{i,j=1}^m
 \Bigl[
  S(\xi_i,\zeta_j)
  +\bigl(U(\xi_i,\zeta_j)+\widetilde V(\xi_i,\zeta_j)\bigr)\\[-2pt]
 &\hspace{9em}\cdot
  \bigl(
   g_{h^{(m)}}(\xi_i,a(\xi_i))
   +g_{h^{(m)}}(\zeta_j,b(\zeta_j))
  \bigr)
 \Bigr].
\end{aligned}
\]
Taking expectation over the representatives turns the last sum into
\(\Gamma(g_{h^{(m)}},a,b)\).  Thus
\(\MF_m(\boldsymbol h)\le\Gamma(g_{h^{(m)}},a,b)\) for every
\((a,b)\in\mathcal T\); taking the infimum proves the second inequality.
\end{proof}

\begin{theorem}
\label{thm:finite-guarantee}
For every \(m\) and every \(m\)-level score function \(h^{(m)}\)
whose heights satisfy~\eqref{eq:finite-scores}, and for every instance
\(G=(L,R,w,E)\), \(\HR(h^{(m)})\) satisfies
\[
 \E[\ALG(G)]\ge \MC_m(\boldsymbol h)\,\OPT(G)
             =\MF_m(\boldsymbol h)\,\OPT(G).
\]
\end{theorem}

\begin{proof}
Proposition~\ref{prop:fixed-score-guarantee} gives
\[
 \E[\ALG(G)]
 \ge\gamma_{\mathrm{EW}}(h^{(m)})\,\OPT(G).
\]
Now apply Proposition~\ref{prop:step-flow-bridge} and the identity
\(\MF_m(\boldsymbol h)=\MC_m(\boldsymbol h)\) from
Proposition~\ref{prop:finite-cut}.
\end{proof}

\subsection{Why the formulation is polynomial}

For an \(m\)-level score function, each curve has
\((m+1)^m\) possible threshold
vectors, so a direct enumeration of the pair has
\((m+1)^{2m}\) cases.  For a fixed step profile, the minimum-cut network
represents all of them with \(O(m^2)\) nodes and \(O(m^3)\) arcs.
For a fixed table \(\boldsymbol q\), every \(q\)-dependent capacity is
linear in \(\boldsymbol q\), and the resulting flow problem is a
linear program.  The outer map from scores
to that table remains nonlinear, however, through
\(g_{ij}=h_j/(h_i+h_j)\) and
\(q_{ij}=g_{ij}-g_{i,j+1}\).  We use numerical search only to choose a
candidate step profile, and then certify that fixed parameter choice
exactly.

\section[A 240-Step Score and Exact-Integer Verification]
{A \(240\)-Step Score and Exact-Integer Verification}
\label{sec:certificate}

Section~\ref{sec:finite} constructed a finite-flow lower bound for
every step score.  We now specify a \(240\)-step score function
\(h^\star\) and a reproducible exact-integer maximum-flow computation
for its network. A \href{https://github.com/ahqspxy/harmonic-ranking-certificate}
{complete verification artifact} is publicly available.

Positive rescaling changes neither the algorithm nor its gain-sharing rule:
for every \(c>0\),
\[
 \pi_{uv}^{\,ch}=c\,\pi_{uv}^{\,h},
 \qquad
 g_{ch}=g_h.
\]
We therefore normalize the first height to one when displaying the
score.

Set \(m=240\).  The supplementary file
\path{finite_rank_scores_n240.txt} contains exactly 240 integers
\(H_1>\cdots>H_{240}>0\).  Write
\(\boldsymbol H:=(H_1,\ldots,H_{240})\); these are the step heights of
\[
 h^\star(x)=H_i
 \quad\text{for }x\in[(i-1)/240,i/240),
 \qquad h^\star(1)=0.
\]

\begin{figure}[tbp]
\centering
\begin{tikzpicture}
\begin{axis}[
 width=.72\linewidth,
 height=.36\linewidth,
 xmin=0,xmax=1,
 ymin=0,ymax=1.02,
 xlabel={rank \(x\)},
 ylabel={normalized height},
 grid=major,
]
\addplot[
 thresholdblue,
 thick,
 no marks,
 const plot,
]
table[
 skip first n=1,
 x expr={(\coordindex+.5)/240},
 y expr={\thisrowno{0}/1e30},
] {supplement/finite_rank_scores_n240.txt};
\end{axis}
\end{tikzpicture}
\caption{The normalized heights \(H_i/H_1\), plotted at their cell
midpoints.  The algorithm uses the corresponding \(240\)-step
piecewise-constant score.}
\label{fig:certified-score}
\end{figure}
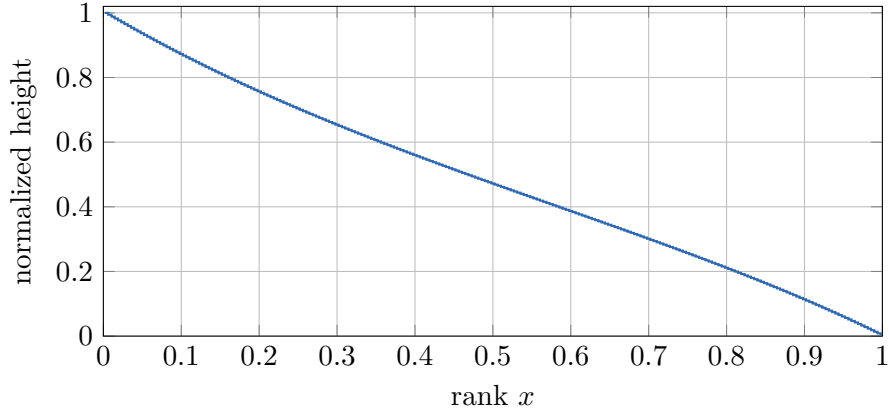

The numerical search was guided by the three-parameter profile
\begin{equation}\label{eq:smooth-score-guide}
 \widetilde h(x)
 :=
 (1-x)^p
 \exp\!\left(
  c_1\bigl((1-x)-1\bigr)
  +c_2\bigl((1-x)^2-1\bigr)
 \right),
\end{equation}
with
\[
 p=0.9938,\qquad
 c_1=-0.74112772,\qquad
 c_2=0.58515414.
\]
Using the sampling offset \(\delta_{\mathrm{samp}}=0.30\), the checked
integers are generated by
\[
 H_i
 =
 \left\lfloor
 10^{30}
 \frac{
  \widetilde h((i-1+\delta_{\mathrm{samp}})/m)
 }{
  \widetilde h(\delta_{\mathrm{samp}}/m)
 }
 \right\rfloor
 \qquad(i\in[m]).
\]
This smooth formula explains only how the candidate was selected.  The
proof uses the checked integer list: the artifact verifier confirms
its count, positivity, and monotonicity.  Put \(H_{m+1}=0\) only for
the telescoping formulas.
The capacities are computed from the exact identities
\begin{align}
 q_{ij}
 &=
 \frac{H_i(H_j-H_{j+1})}
 {(H_i+H_j)(H_i+H_{j+1})},\label{eq:exact-q}\\
 r_{ij}
 &=
 \frac{
 (m-j+1)H_j(H_i+H_{j+1})
 -(m-j)H_{j+1}(H_i+H_j)}
 {(H_i+H_j)(H_i+H_{j+1})}.\label{eq:exact-r}
\end{align}
The verifier independently recomputes all
\(2m^2=115{,}200\) stored \(q\)- and \(r\)-capacities from these
identities, including the \(m^{-2}\) normalization in
\eqref{eq:finite-network}, and checks that the declared integer scale
is \(10^{12}\).
For every objective arc of exact rational capacity \(c_e\), the checker
uses integer capacity
\[
 \underline c_e:=\lfloor 10^{12}c_e\rfloor.
\]
Every hard implication arc receives the scaled big-\(M\) capacity
\(2\cdot10^{12}\).
The source-only cut and the cut placing every nonterminal node on the
source side each have normalized capacity at most one, so a minimum
cut below one never crosses a hard arc.

\begin{center}
\begin{tabular}{@{}lr@{}}
\toprule
grid size \(m\) & \(240\)\\
step heights & \(240\) positive integers\\
stored rounded \(q/r\) capacities & \(115{,}200\)\\
network nodes & \(115{,}202\)\\
forward arcs & \(27{,}993{,}120\)\\
integer scale & \(10^{12}\)\\
computed flow value & \(698015475248\)\\
source-reachable cut value & \(698015475248\)\\
implied factor & \(0.698015475248\)\\
\bottomrule
\end{tabular}
\end{center}

\begin{lemma}
\label{lem:rounding}
If the rounded integer network has a feasible flow of value \(F\), then
\[
 \MC_{240}(\boldsymbol H)=\MF_{240}(\boldsymbol H)
 \ge \frac{F}{10^{12}}.
\]
\end{lemma}

\begin{proof}
Every rounded objective capacity is no larger than its exact capacity,
while each scaled big-\(M\) capacity is exact.  After division by
\(10^{12}\), any feasible flow in the rounded network is therefore
feasible in the exact rational network.
\end{proof}

\begin{theorem}[Exactly verified \(240\)-step score function]
\label{thm:certificate}
The supplied step function \(h^\star\) satisfies
\[
 \MF_{240}(\boldsymbol H)=\MC_{240}(\boldsymbol H)
 \ge\frac{698015475248}{10^{12}}
 =0.698015475248>0.698.
\]
\end{theorem}

\begin{proof}
The integer checker exhibits a feasible flow and the source-reachable
cut in its terminal residual graph, both of value
\[
 \texttt{flow\_integer}
 =
 \texttt{cut\_integer}
 =
 698015475248.
\]
Thus max-flow/min-cut duality verifies this value exactly for the
rounded network, and Lemma~\ref{lem:rounding} transfers its feasible
flow to the exact rational network.  The equal-valued cut verifies
optimality but is not needed for the lower-bound direction.
\end{proof}

The supplementary README records the artifact hashes, exact commands,
expected output, resource requirements, and implementation details.

\begin{proof}[Proof of Theorem~\ref{thm:main}]
Apply Theorem~\ref{thm:finite-guarantee} to
\(\HR(h^\star)\), and use Theorem~\ref{thm:certificate}.  This gives
\[
 \E[\ALG(G)]
 \ge 0.698015475248\,\OPT(G)
\]
for every \(G\).
\end{proof}

\section{Online Harmonic Ranking}\label{sec:online}

In the vertex-weighted online model, every offline vertex \(u\in L\) has a
known weight \(w_u\), and every edge incident to \(u\) has that weight.
Online vertices arrive in a uniformly random order and reveal their
neighborhoods on arrival.  After discarding zero-valued offline
vertices, we assume \(w_u>0\).  We also assume that \(n=|R|\) is known
in advance.

For the coupling below, fix an arbitrary order $\prec_L$ on $L$ and
use the following admissible specialization of
Section~\ref{sec:tie-breaking}: among equal-priority pairs, order first
by increasing rank of the $R$-endpoint and then by $\prec_L$, with
any remaining ties broken by a fixed order. This rule is admissible:
lowering a rank weakly improves every incident pair, while comparisons
among nonincident pairs are unchanged.

\begin{algorithmbox}{Online Harmonic Ranking}
Independently sample an offline rank
$r_u\sim\Unif[0,1)$ for every $u\in L$. Independently sample $n$ uniform
timestamps with order statistics
\[
 t_{(1)}<t_{(2)}<\cdots<t_{(n)}.
\]
Assign the \(k\)-th arriving online vertex \(v\) the timestamp
\(t_v:=t_{(k)}\).  Match \(v\), if possible, to an available neighbor
\(u\) maximizing
\begin{equation}\label{eq:online-harmonic-gain}
 w_u g_{h^\star}(t_v,r_u)
 =
 w_u\frac{h^\star(r_u)}
 {h^\star(t_v)+h^\star(r_u)},
\end{equation}
breaking ties by \(\prec_L\).
\end{algorithmbox}

\begin{proof}[Proof of Corollary~\ref{cor:online}]
Fix the graph, the arrival order, and all sampled ranks and timestamps.
For a potential pair \(uv\), its Harmonic Ranking priority is
\[
 \pi_{uv}
 =w_u\frac{h^\star(t_v)h^\star(r_u)}
 {h^\star(t_v)+h^\star(r_u)}.
\]
For fixed \(v\), the positive factor \(h^\star(t_v)\) is common to all
neighbors.  Hence maximizing~\eqref{eq:online-harmonic-gain} is exactly
maximizing \(\pi_{uv}\).

Under the identification $x_u=r_u$ and $x_v=t_v$, the above
admissible specialization orders equal numerical priorities first by
increasing $t_v$. Within a fixed arrival $v$, this first key is
constant, so tied pairs are ordered exactly by $\prec_L$, as in
Online Harmonic Ranking. For fixed \(u\), \(\pi_{uv}\) is nonincreasing in \(t_v\),
with equality resolved in favor of the earlier arrival.  Thus, if a
pair of a later arrival precedes a pair of an earlier arrival, the two
pairs are vertex-disjoint.  A stable sort by arrival time swaps only
adjacent pairs of this kind.  Their probes commute, so the sort groups
the scan into arrival blocks without changing the matching. 
Within each block, the scan is precisely Online Harmonic Ranking,
since both executions use the same order $\prec_L$ on tied offline
neighbors; unrealized pairs and probes after the online vertex is
matched are harmless no-ops.

Thus the online matching is pathwise equal to \(\HR(h^\star)\) on the
edge-weighted instance \(w_{uv}=w_u\).  Assigning the order statistics
of independent uniforms to a uniformly random arrival order makes the
labeled timestamps independent and uniform, also independently of the
offline ranks.  Theorem~\ref{thm:main} therefore gives the claimed
\(0.698015475248\) ratio.
\end{proof}

The online guarantee above uses the score-balanced gain-sharing rule
\(g_{h^\star}\), because it is inherited from the preceding
edge-weighted analysis.  The vertex-weighted variational benchmark
considered below, however, allows the full class of complementary
monotone gain-sharing rules.  We record this distinction next.

\subsection{Complementary gain-sharing rules and the variational hierarchy}
\label{sec:variational-hierarchy}

Let \(\mathfrak G_{\mathrm{comp}}\) be
the class of measurable
gain-sharing rules \(g:[0,1]^2\to[0,1]\) that are nondecreasing in the first
coordinate, nonincreasing in the second, and satisfy
\[
 g(x,y)+g(y,x)=1,\qquad g(x,1)=0\quad(x<1),\qquad g(1,y)=1\quad(y<1).
\]
Let
\[
 \mathfrak H
 :=
 \{g_h:h\text{ satisfies the score conditions in
 Section~\ref{sec:model}}\}.
\]
Here every \(g_h\) carries the boundary
extension~\eqref{eq:g-boundary-extension}.
Then \(\mathfrak H\subseteq\mathfrak G_{\mathrm{comp}}\).  Define
\begin{align*}
 \mathcal T_{\mathrm{EW}}&:=\mathcal T,\\
 \mathcal T_{\mathrm{VW}}
 &:=\{(a,b)\in\mathcal T:
       b\text{ has a nondecreasing representative}\},\\
 \mathcal T_{\mathrm{UW}}
 &:=\{(a,b)\in\mathcal T:
       a,b\text{ have nondecreasing representatives}\}.
\end{align*}
Thus
\(\mathcal T_{\mathrm{UW}}\subseteq
 \mathcal T_{\mathrm{VW}}\subseteq
 \mathcal T_{\mathrm{EW}}\).

\begin{table}[H]
\centering
\begin{tabular}{@{}ccl@{}}
\toprule
model&gain-sharing rule&cutoff geometry\\
\midrule
edge-weighted (EW)
 &\(g=g_h\)&\(a,b\) arbitrary\\
vertex-weighted (VW)
 &\(g\in\mathfrak G_{\mathrm{comp}}\)
 &\(b\) nondecreasing, \(a\) arbitrary\\
unweighted (UW)
 &\(g\in\mathfrak G_{\mathrm{comp}}\)
 &\(a,b\) nondecreasing\\
\bottomrule
\end{tabular}
\caption{Gain-sharing classes and cutoff geometry in the three matching
models.}
\label{tab:model-hierarchy}
\end{table}

The common functional~\eqref{eq:Gamma} gives
\begin{equation}\label{eq:model-hierarchy}
\begin{aligned}
 \gamma_{\mathrm{EW}}
 &=
 \sup_{g\in\mathfrak H}
 \inf_{(a,b)\in\mathcal T_{\mathrm{EW}}}\Gamma(g,a,b),\\
 \gamma_{\mathrm{VW}}
 &:=
 \sup_{g\in\mathfrak G_{\mathrm{comp}}}
 \inf_{(a,b)\in\mathcal T_{\mathrm{VW}}}\Gamma(g,a,b),\\
 \gamma_{\mathrm{UW}}
 &:=
 \sup_{g\in\mathfrak G_{\mathrm{comp}}}
 \inf_{(a,b)\in\mathcal T_{\mathrm{UW}}}\Gamma(g,a,b).
\end{aligned}
\qquad
\gamma_{\mathrm{EW}}
\le\gamma_{\mathrm{VW}}
\le\gamma_{\mathrm{UW}}.
\end{equation}
The first inequality uses both
\(\mathfrak H\subseteq\mathfrak G_{\mathrm{comp}}\) and
\(\mathcal T_{\mathrm{VW}}\subseteq\mathcal T_{\mathrm{EW}}\);
the second uses
\(\mathcal T_{\mathrm{UW}}\subseteq\mathcal T_{\mathrm{VW}}\).
These are structural factor-revealing values.

\subsection{Exact recovery of the Mahdian--Yan finite-grid value}
\label{sec:correspondences}

The variational hierarchy above is formulated in the language of
complementary gain sharing and cutoff geometry.  At its unweighted
endpoint, the natural classical benchmark is the strongly
factor-revealing program underlying the previous \(0.696\) analysis
of random-arrival Ranking by Mahdian and Yan.  The two descriptions
appear quite different: our formulation is a max--min game over
complementary gain-sharing tables and pairs of monotone cutoff paths,
whereas the Mahdian--Yan formulation is a polynomial-size linear
program with cumulative rows and transpose-balance constraints.  The
following result shows that their optimal values agree exactly at
every finite grid resolution.

For every grid size \(m\),
Appendix~\ref{app:grid-correspondence} defines the exact whole-cell
unweighted value \(\gamma_{\mathrm{UW},m}\), its reflected
complementary two-path value \(V_m\), and the Mahdian--Yan
program \(\polyLPp(m)\).

\begin{theorem}[Exact finite-grid value correspondence]
\label{thm:MY-equivalence}
For every \(m\ge1\),
\begin{equation}
  \gamma_{\mathrm{UW},m}
  =
  V_m
  =
  \operatorname{val}\bigl(\polyLPp(m)\bigr).
  \label{eq:MY-equality}
\end{equation}
\end{theorem}

Equation~\eqref{eq:MY-equality} is exact for each fixed \(m\); it is
neither an asymptotic statement nor a numerical coincidence.  At the
level of optimal certificate values, it shows that the complementary
primal--dual formulation loses no strength relative to the
Mahdian--Yan program at the unweighted endpoint.  Equivalently,
\(\polyLPp(m)\) provides a polynomial-size value characterization of
the exact whole-cell cutoff game.  In this precise sense, the
variational hierarchy developed here recovers the full finite-grid
strength of the classical unweighted analysis rather than defining a
parallel relaxation with a merely similar value.

The first equality is proved in
Appendix~\ref{app:grid-correspondence} by reflecting both grid
coordinates and both cutoff paths, followed by an overlap-removal
argument.  The second equality is proved in
Appendix~\ref{app:balance-price}: transpose balance is partially
dualized into a complementary price table, fixed-price contraction
identifies the inner problem with the two-path game, and
monotone-price normalization restricts the price to the required
monotone class.  The outer gain-sharing optimization
ranges over the grid discretization of the full class
\(\mathfrak G_{\mathrm{comp}}\), rather than over the score-balanced
subclass; the theorem does not assert that the optimum is attained
within that subclass.





\section{Discussion}\label{sec:discussion}

In this work, we introduce the Harmonic Ranking algorithm and prove a
competitive ratio of \(0.698\). We conclude with several observations that
emerged from our investigation but are not needed for the main results.

In Section~\ref{sec:online}, we formalize a hierarchy of factor-revealing problems
corresponding to the edge-weighted, vertex-weighted, and unweighted
settings:
\[
\gamma_{\mathrm{EW}}
\le
\gamma_{\mathrm{VW}}
\le
\gamma_{\mathrm{UW}}.
\]
Our computer-assisted analysis shows that
\(\gamma_{\mathrm{EW}}>0.698\), while \citet{PengTang2025} showed that
\(\gamma_{\mathrm{UW}}<0.703\). Thus, although the remaining quantitative
gap is small, it is natural to ask for the exact values of these three
quantities and, in particular, whether the inequalities in the hierarchy
are strict.
Our numerical experiments suggest a rather delicate picture.

\paragraph{Edge-weighted versus vertex-weighted.}
As discussed in Section~\ref{sec:finite}, if we relax the gain-sharing functions from the score-balanced family \(\mathfrak H\) to the more general family
\(\mathfrak G_{\mathrm{comp}}\), then, after discretization, the resulting
max-flow formulation becomes a linear program. Numerically, imposing the
additional monotonicity constraint on one cutoff curve---the structural
property available in the vertex-weighted setting---does not appear to
improve the optimal value. More precisely, our computations suggest
\[
\sup_{g\in\mathfrak H}
 \inf_{(a,b)\in\mathcal T_{\mathrm{EW}}}\Gamma(g,a,b)
\;\le\;
\sup_{g\in\mathfrak G_{\mathrm{comp}}}
 \inf_{(a,b)\in\mathcal T_{\mathrm{EW}}}\Gamma(g,a,b)
\;\overset{!}{=}\;
\sup_{g\in\mathfrak G_{\mathrm{comp}}}
 \inf_{(a,b)\in\mathcal T_{\mathrm{VW}}}\Gamma(g,a,b),
\]
where \(\overset{!}{=}\) denotes a numerical observation rather than a
proved equality.

At the same time, we see no structural reason for the optimal gain-sharing
function in the larger family \(\mathfrak G_{\mathrm{comp}}\) to be
score-balanced. We conjecture that the first inequality above is strict. 
If so, the distinction between the
edge-weighted and vertex-weighted analyses would arise not from
the additional monotonicity of the cutoff curve, but from its interaction
with the admissible class of gain-sharing functions.

\paragraph{Vertex-weighted versus unweighted.}
For both the vertex-weighted and unweighted settings, the discretized
factor-revealing problems can be formulated as linear programs. At every
grid size we tested, the corresponding optimal values are distinct, with
the unweighted program attaining the larger value. It remains unclear,
however, whether this separation persists in the continuum limit or
whether the two sequences converge to the same value as the grid becomes
arbitrarily fine.

The numerical scale of these possible separations is reminiscent of the
fine distinctions among the approximation thresholds for problems such as
MAX 2-AND, MAX DI-CUT, and MAX CUT~\citep{BrakensiekHPZ26}. We do not intend to suggest
that the hierarchy studied here has comparable significance. In
particular, our hierarchy concerns the power of a specific
factor-revealing, primal--dual analysis rather than a separation between
the underlying matching problems themselves. Nevertheless, we find the
phenomenon conceptually intriguing: seemingly modest changes in the
structural information available to the analysis may lead to distinct
optimal constants, even when the differences occur only at the third
decimal place. We record these observations both as directions for future
work and as potentially useful clues toward understanding the limits of
this analysis framework.

\section*{AI Disclosure}

The \HR{} algorithm and the Mutual Proposals framework were conceived
entirely by the authors, without AI assistance.  The minimum-cut and
maximum-flow formulation was developed through iterative discussions
between the authors and OpenAI's GPT-5.6 Sol.  Prompted by a numerical
coincidence between values reported by Peng and Tang and the
Mahdian--Yan bound, the authors strongly conjectured that the variational
problem admits a polynomial-size representation.  The proof of the
exact correspondence with the Mahdian--Yan program reported in the
appendix was subsequently completed independently by GPT.

The final score function \(h^\star\) was constructed through numerical
experimentation.  The authors initially proposed \(h(x)=1-x\), and early
experiments strongly suggested that a ratio of \(\ln 2\approx0.693\)
was attainable.  Experiments with a general complementary gain-sharing rule \(g\) then produced ratios near \(0.698\).  GPT subsequently
proposed the final \(240\)-step function by fitting the numerical
results.  The authors performed the early numerical experiments; as the
capabilities
of available AI systems improved, the subsequent numerical exploration
was soon taken over by AI.

The overall architecture of the paper, including the organization of
its arguments, was designed by the authors.  AI assistance was used
extensively to develop technical details and fill in substantial parts
of the exposition.  The authors have reviewed the final manuscript and
assume full responsibility for all definitions, claims, proofs,
computations, and presentation.

\clearpage

\bibliographystyle{plainnat}
\bibliography{edge_weighted_oblivious_matching}

\clearpage
\appendix


\section{Proof of the Unweighted Grid Correspondence}
\label{app:grid-correspondence}

This appendix proves the unweighted correspondence stated in
Theorem~\ref{thm:MY-equivalence}.  The proof is independent of the edge-weighted
competitive-ratio analysis.

We work throughout on the finite grid of Section~\ref{sec:finite}.  In the
unweighted specialization, both threshold vectors are nondecreasing,
and the gain rule may range over all complementary monotone tables.
The proof has two conceptual steps.  First, reflecting both grid
coordinates turns the whole-cell game into a complementary two-path
game.  Second, Appendix~\ref{app:balance-price} identifies that
two-path game with the Mahdian--Yan program
\(\polyLPp(m)\).

Fix \(m\ge1\).  Let \(\mathfrak G_{\mathrm{comp},m}\) be the set of
tables
\[
  \widehat g=(\widehat g_{ij})_{i,j\in[m]}\in[0,1]^{m\times m}
\]
that are nondecreasing in \(i\), nonincreasing in \(j\), and satisfy
\[
  \widehat g_{ij}+\widehat g_{ji}=1
  \qquad(i,j\in[m]),
\]
with the appended boundary values
\(\widehat g_{i,m+1}:=0\) and \(\widehat g_{m+1,j}:=1\).

Let
\[
  \mathcal T_{\mathrm{UW},m}
  :=
  \{(\boldsymbol a,\boldsymbol b)\in\mathcal T_m:
    \boldsymbol a,\boldsymbol b
    \text{ are nondecreasing}\}.
\]
For \((\boldsymbol a,\boldsymbol b)\in\mathcal T_{\mathrm{UW},m}\),
define
\begin{equation}
  \widehat\Gamma_m(\widehat g,\boldsymbol a,\boldsymbol b)
  :=
  \frac1{m^2}\sum_{i,j=1}^m
  \left[
    S_{ij}
    +(U_{ij}+\widetilde V_{ij})
    \bigl(
      \widehat g_{i,a_i+1}
      +\widehat g_{j,b_j+1}
    \bigr)
  \right],
  \label{eq:finite-unweighted-Gamma}
\end{equation}
and
\begin{equation}
  \gamma_{\mathrm{UW},m}
  :=
  \max_{\widehat g\in\mathfrak G_{\mathrm{comp},m}}
  \min_{(\boldsymbol a,\boldsymbol b)\in
       \mathcal T_{\mathrm{UW},m}}
  \widehat\Gamma_m(\widehat g,\boldsymbol a,\boldsymbol b).
  \label{eq:finite-unweighted-value}
\end{equation}

\subsection{The complementary two-path game}

Let
\[
  \mathcal B_m
  :=
  \left\{
    \boldsymbol a=(a_0,\ldots,a_m)
    \in\{0,\ldots,m\}^{m+1}:
    a_0\le\cdots\le a_m=m
  \right\}.
\]
For \(\boldsymbol a\in\mathcal B_m\) and
\(j\in\{0,\ldots,m-1\}\), put
\[
  a_j^-:=\min\{k\in\{0,\ldots,m\}:a_k>j\},
  \qquad
  \ell_{\boldsymbol a}(i)
  :=
  1-\frac{a_{i-1}^-}{m}
  \quad(i\in[m]).
\]

Let \(\mathcal K_m\) consist of the tables \(g(i,r)\),
\(i,r\in[m]\), with the implicit boundary \(g(i,0):=0\), satisfying
\[
  g(i,1)\ge0,\qquad
  g(i,r)\le g(i,r+1)\quad(r<m),
  \qquad
  g(i,r)+g(r,i)=1.
\]
Complementarity and monotonicity imply \(0\le g\le1\) and that
\(g\) is nonincreasing in its first coordinate.

For \(\boldsymbol a,\boldsymbol b\in\mathcal B_m\), define
\begin{equation}
\begin{aligned}
  \Phi_g(\boldsymbol a,\boldsymbol b)
  :={}&1
  +\frac1{m^2}\sum_{i=1}^m
  \Bigl[
    a_{i-1}\bigl(g(i,a_{i-1})-1\bigr)
    +b_{i-1}\bigl(g(i,b_{i-1})-1\bigr)
  \Bigr]\\
  &+\frac1m\sum_{i=1}^m
  \Bigl[
    \ell_{\boldsymbol b}(i)g(i,a_{i-1})
    +\ell_{\boldsymbol a}(i)g(i,b_{i-1})
  \Bigr],
  \label{eq:Phi}
\end{aligned}
\end{equation}
and
\begin{equation}
  V_m
  :=
  \max_{g\in\mathcal K_m}
  \min_{\boldsymbol a,\boldsymbol b\in\mathcal B_m}
  \Phi_g(\boldsymbol a,\boldsymbol b).
  \label{eq:Vm}
\end{equation}

Under the reflection used below, the whole-cell gain-sharing table
\(\widehat g\) and cutoff pair
\((\boldsymbol a,\boldsymbol b)\in\mathcal T_{\mathrm{UW},m}\)
are mapped to the complementary gain-sharing table \(g\in\mathcal K_m\) and a
reflected path pair
\((\boldsymbol c,\boldsymbol d)\in\mathcal B_m^2\), respectively.
The reflected pair initially satisfies the additional
transpose-overlap-free condition stated in Lemma~\ref{lem:direct-whole-cell-reflection}.

\begin{lemma}[Direct whole-cell reflection]
\label{lem:direct-whole-cell-reflection}
For \(\widehat g\in\mathfrak G_{\mathrm{comp},m}\), define
\[
  (\mathscr R\widehat g)(k,r)
  :=
  \widehat g_{m+1-k,m+1-r},
  \qquad
  (\mathscr R\widehat g)(k,0):=0.
\]
Then \(\mathscr R\) is a bijection from
\(\mathfrak G_{\mathrm{comp},m}\) onto \(\mathcal K_m\).

Moreover, if \(g=\mathscr R\widehat g\), then
\[
  \min_{(\boldsymbol a,\boldsymbol b)\in
       \mathcal T_{\mathrm{UW},m}}
  \widehat\Gamma_m(\widehat g,\boldsymbol a,\boldsymbol b)
  =
  \min_{(\boldsymbol c,\boldsymbol d)\in\mathcal P_m^\perp}
  \Phi_g(\boldsymbol c,\boldsymbol d),
\]
where
\[
  \mathcal P_m^\perp
  :=
  \left\{
    (\boldsymbol c,\boldsymbol d)\in\mathcal B_m^2:
    X_{kr}Y_{rk}=0
    \text{ for all }k,r\in[m]
  \right\},
\]
with
\(X_{kr}:=\one_{\{c_{k-1}\ge r\}}\) and
\(Y_{kr}:=\one_{\{d_{k-1}\ge r\}}\).
\end{lemma}

\begin{proof}
Reversing both indices turns monotonicity of \(\widehat g\) in its
second coordinate into monotonicity of \(g(k,r)\) in \(r\).
Complementarity is preserved, and
\(\widehat g_{i,m+1}=0\) becomes \(g(k,0)=0\).
The same index reversal is its own inverse, so \(\mathscr R\) is a bijection.

Fix
\((\boldsymbol a,\boldsymbol b)\in\mathcal T_{\mathrm{UW},m}\)
and define reflected paths by
\begin{equation}
  c_{k-1}:=m-a_{m+1-k},
  \qquad
  d_{k-1}:=m-b_{m+1-k}
  \quad(k\in[m]),
  \qquad
  c_m=d_m=m.
  \label{eq:direct-reflected-paths}
\end{equation}
Because \(\boldsymbol a\) and \(\boldsymbol b\) are nondecreasing,
\(\boldsymbol c,\boldsymbol d\in\mathcal B_m\).

Under the index change
\(i=m+1-k\), \(j=m+1-r\),
\[
  U_{ij}=X_{kr},
  \qquad
  \widetilde V_{ij}=Y_{rk}.
\]
Hence the whole-cell nonoverlap condition
\(U_{ij}\widetilde V_{ij}=0\) becomes
\(X_{kr}Y_{rk}=0\).  Conversely,
\[
  a_i=m-c_{m-i},
  \qquad
  b_i=m-d_{m-i},
\]
recovers the original threshold vectors.  Thus
\eqref{eq:direct-reflected-paths} bijects
\(\mathcal T_{\mathrm{UW},m}\) with \(\mathcal P_m^\perp\).

Write
\[
  P_k:=g(k,c_{k-1}),
  \qquad
  Q_r:=g(r,d_{r-1}).
\]
The reflection gives
\(\widehat g_{i,a_i+1}=P_k\) and
\(\widehat g_{j,b_j+1}=Q_r\).
Since \(X_{kr}Y_{rk}=0\), the middle-state indicator is
\(1-X_{kr}-Y_{rk}\).  Therefore
\[
  m^2\widehat\Gamma_m(\widehat g,\boldsymbol a,\boldsymbol b)
  =
  m^2-\sum_{k,r}(X_{kr}+Y_{rk})
  +\sum_{k,r}(X_{kr}+Y_{rk})(P_k+Q_r).
\]
The required counts are
\[
  \sum_r X_{kr}=c_{k-1},
  \qquad
  \sum_k Y_{rk}=d_{r-1},
\]
and
\[
  \sum_r Y_{rk}=m\ell_{\boldsymbol d}(k),
  \qquad
  \sum_k X_{kr}=m\ell_{\boldsymbol c}(r).
\]
Expanding the four gain terms now gives
\[
  \widehat\Gamma_m(\widehat g,\boldsymbol a,\boldsymbol b)
  =
  \Phi_g(\boldsymbol c,\boldsymbol d),
\]
which proves the claim.
\end{proof}

\begin{lemma}[Transpose-overlap removal]
\label{lem:uw-overlap-removal}
For every \(g\in\mathcal K_m\),
\[
  \min_{(\boldsymbol c,\boldsymbol d)\in\mathcal P_m^\perp}
  \Phi_g(\boldsymbol c,\boldsymbol d)
  =
  \min_{\boldsymbol c,\boldsymbol d\in\mathcal B_m}
  \Phi_g(\boldsymbol c,\boldsymbol d).
\]
\end{lemma}

\begin{proof}
Write
\(\bar c_i:=c_{i-1}\), \(\bar d_i:=d_{i-1}\), and define
\[
  L_c(r):=|\{j\in[m]:\bar c_j\ge r\}|,
  \qquad
  L_d(r):=|\{j\in[m]:\bar d_j\ge r\}|.
\]
Since \(L_c(r)=m\ell_{\boldsymbol c}(r)\) and similarly for
\(\boldsymbol d\), multiplying
\eqref{eq:Phi} by \(m^2\) gives
\begin{equation}
\begin{aligned}
  m^2\Phi_g(\boldsymbol c,\boldsymbol d)
  ={}&m^2-\sum_i(\bar c_i+\bar d_i)\\
  &+\sum_i
  \Bigl[
    (\bar c_i+L_d(i))g(i,\bar c_i)
    +(\bar d_i+L_c(i))g(i,\bar d_i)
  \Bigr].
  \label{eq:direct-compact-Phi}
\end{aligned}
\end{equation}

Suppose \(X_{i_0t}Y_{ti_0}=1\).  Put
\(s:=\bar c_{i_0}\ge t\), and let \(i\) be the first index in the
plateau on which \(\bar c_i=s\).  Since \(\boldsymbol d\) is
nondecreasing,
\[
  \bar d_s\ge\bar d_t\ge i_0\ge i.
\]
Replace \(\bar c_i=s\) by \(s-1\).  Monotonicity of
\(\boldsymbol c\) is preserved, at least one transpose overlap is
removed, and no new overlap is created.

Let
\(\Delta_{is}:=g(i,s)-g(i,s-1)\ge0\).
By \eqref{eq:direct-compact-Phi}, the change in \(m^2\Phi_g\), new value
minus old value, is
\begin{align*}
  &1-g(i,s-1)
   -(s+L_d(i))\Delta_{is}
   -g(s,\bar d_s)\\
  &\qquad
   =g(\bar d_s,s)-g(i,s-1)
    -(s+L_d(i))\Delta_{is}\\
  &\qquad
   \le (1-s-L_d(i))\Delta_{is}
   \le0.
\end{align*}
The equality uses complementarity.  For the first inequality,
\(g\) is nonincreasing in its first coordinate and
\(\bar d_s\ge i\), hence
\[
  g(\bar d_s,s)
  \le g(i,s)
  =g(i,s-1)+\Delta_{is}.
\]

Repeating the operation terminates because
\(\sum_i\bar c_i\) strictly decreases.  Thus every pair in
\(\mathcal B_m^2\) can be transformed into one in
\(\mathcal P_m^\perp\) of no larger value.  Since
\(\mathcal P_m^\perp\subseteq\mathcal B_m^2\), the two minima are
equal.
\end{proof}

\begin{corollary}
\label{cor:uw-kernel-value}
\[
  \gamma_{\mathrm{UW},m}=V_m.
\]
\end{corollary}

\begin{proof}
Lemma~\ref{lem:direct-whole-cell-reflection} identifies the two table classes, and
Lemma~\ref{lem:uw-overlap-removal} removes the only additional
restriction created by the reflection.  Maximizing over the
corresponding gain-sharing
tables gives the result.
\end{proof}

\subsection{The Mahdian--Yan program}

For completeness, we state the program used in
Theorem~\ref{thm:MY-equivalence}.
The index \(\ell\) is the cumulative active-rank row, \(r\) is the
passive rank, and \(p\) is an auxiliary state rank.  Let
\(x(\ell,r,p)\ge0\), and define
\[
  y(\ell,r,p)
  :=
  \sum_{h=1}^{\ell}x(h,r,p),
  \qquad
  W_{\ell r}
  :=
  \sum_{p=1}^m x(\ell,r,p).
\]
The program \(\polyLPp(m)\) is
\begin{equation}
  \text{minimize}\qquad
  \frac1m\sum_{\ell,r,p=1}^m x(\ell,r,p)
  \label{eq:MY-obj}
\end{equation}
subject to
\begin{align}
  y(\ell,r,\ell)+y(r,\ell,p)
  &\ge\frac1m
  &&(\ell,r,p\in[m]),
  \label{eq:MY-P1}\\
  y(\ell+1,r,p+1)
  &\ge y(\ell,r,p)
  &&(p\le\ell<m),
  \label{eq:MY-P2}\\
  y(\ell,r,p)
  &=y(\ell,r,\ell+1)
  &&(\ell<p),
  \label{eq:MY-P3}\\
  y(\ell+1,r,p)
  &\ge y(\ell,r,\ell+1)
  &&(p\le\ell<m),
  \label{eq:MY-P4}\\
  W_{\ell r}
  &=W_{r\ell}
  &&(\ell,r\in[m]).
  \label{eq:MY-P5}
\end{align}
We call \eqref{eq:MY-P1} the \emph{crossed-coverage constraints} and
\eqref{eq:MY-P5} the \emph{transpose-balance constraints}. The notation \(\polyLPp(m)\) follows Mahdian and Yan's original
formulation; we state the program in full here for self-containment.


\paragraph{The remaining balance--price step.}
The reflection argument above does not by itself prove
\[
  V_m=\operatorname{val}\bigl(\polyLPp(m)\bigr).
\]
That identity requires dualizing the transpose-balance equations and
normalizing the resulting complementary price table. These steps are carried out in
Appendix~\ref{app:balance-price}.


\begin{proof}[Proof of Theorem~\ref{thm:MY-equivalence}]
Corollary~\ref{cor:uw-kernel-value} gives $\gamma_{\mathrm{UW},m}=V_m$.
Corollary~\ref{cor:bp-balance-price} gives
$V_m=\operatorname{val}\bigl(\polyLPp(m)\bigr)$.
Combining the two identities proves
\[
  \gamma_{\mathrm{UW},m}
  =
  V_m
  =
  \operatorname{val}\bigl(\polyLPp(m)\bigr).
\]
\end{proof}

\begin{remark}[Scope of the correspondence]
Theorem~\ref{thm:MY-equivalence} is an equality of optimum values for
the exact whole-cell unweighted discretization and the Mahdian--Yan program
\(\polyLPp(m)\).  It is not a variable-by-variable isomorphism, and
it does not assert that every optimizer in
\(\mathfrak G_{\mathrm{comp},m}\) has the score-balanced
representation used in Theorem~\ref{thm:main}.
\end{remark}


\section{Balance--Price Duality}
\label{app:balance-price}

This appendix proves the second equality in
Theorem~\ref{thm:MY-equivalence},
\[
  V_m
  =
  \operatorname{val}\bigl(\polyLPp(m)\bigr).
\]
The argument has three steps:
\[
  \text{partial balance dualization}
  \longrightarrow
  \text{fixed-price contraction}
  \longrightarrow
  \text{monotone-price normalization}.
\]

Let \(\mathcal C_m\) be the feasible region defined by
\eqref{eq:MY-P1}--\eqref{eq:MY-P4}, the identities
\[
  y(\ell,r,p)=\sum_{h=1}^{\ell}x(h,r,p),
  \qquad
  x(\ell,r,p)\ge0,
\]
and no transpose-balance equations.  Define
\[
  \mathcal D_m
  :=
  \{F\in\R^{m\times m}:
    F_{\ell r}+F_{r\ell}=1
    \text{ for all }\ell,r\},
\]
and
\begin{equation}
  \mathscr R_m(F)
  :=
  \inf_{(x,y)\in\mathcal C_m}
  \frac2m\sum_{\ell,r,p}
  F_{\ell r}x(\ell,r,p).
  \label{eq:bp-R}
\end{equation}
For an arbitrary real \(F\), this infimum is allowed to be
\(-\infty\).

The monotone price class is
\begin{equation}
  \mathcal A_m
  :=
  \left\{
    F\in\mathcal D_m:
    \begin{array}{l}
      0\le F_{\ell r}\le1,\\
      F_{\ell r}\ge F_{\ell+1,r}\quad(\ell<m),\\
      F_{\ell r}\le F_{\ell,r+1}\quad(r<m)
    \end{array}
  \right\}.
  \label{eq:bp-A}
\end{equation}

\subsection{Partial balance dualization}

\begin{lemma}[Exact partial balance dual]
\label{lem:bp-partial-dual}
\[
  \operatorname{val}\bigl(\polyLPp(m)\bigr)
  =
  \sup_{F\in\mathcal D_m}\mathscr R_m(F).
\]
\end{lemma}

\begin{proof}
For each \(\ell<r\), dualize the independent balance equation
\(W_{\ell r}-W_{r\ell}=0\) with a free multiplier
\(\eta_{\ell r}\).  For notational convenience extend the
multipliers antisymmetrically by
\[
  \eta_{r\ell}:=-\eta_{\ell r},
  \qquad
  \eta_{\ell\ell}:=0.
\]
This introduces no additional degrees of freedom.  The partial
Lagrangian over \(\mathcal C_m\) is
\[
  \mathcal L(x,\eta)
  =
  \frac1m\sum_{\ell,r}W_{\ell r}
  -
  \sum_{\ell<r}\eta_{\ell r}
  (W_{\ell r}-W_{r\ell}).
\]
Since
\[
  \sum_{\ell<r}\eta_{\ell r}
  (W_{\ell r}-W_{r\ell})
  =
  \sum_{\ell,r}\eta_{\ell r}W_{\ell r},
\]
we may write
\[
  \mathcal L(x,\eta)
  =
  \sum_{\ell,r}
  \left(\frac1m-\eta_{\ell r}\right)W_{\ell r}.
\]

Set
\[
  F_{\ell r}:=\frac12-\frac m2\eta_{\ell r}.
\]
Then \(F_{\ell r}+F_{r\ell}=1\), and this parametrizes
\(\mathcal D_m\) bijectively.  Since
\(\frac1m-\eta_{\ell r}=\frac2mF_{\ell r}\) and
\(W_{\ell r}=\sum_p x(\ell,r,p)\),
\[
  \mathcal L(x,\eta)
  =
  \frac2m\sum_{\ell,r,p}
  F_{\ell r}x(\ell,r,p).
\]
Hence
\[
  \inf_{(x,y)\in\mathcal C_m}\mathcal L(x,\eta)
  =
  \mathscr R_m(F).
\]

It remains only to justify that the partial Lagrangian dual is exact.
The deleted constraints are linear equalities in a finite-dimensional
LP.  Maximizing over their free multipliers together with the dual
variables for the constraints defining \(\mathcal C_m\) recovers the
full dual of \(\polyLPp(m)\).  The primal is feasible, for example at
\(x(\ell,r,p)=1\), and its nonnegative objective is bounded below.
LP strong duality therefore gives the claim.
\end{proof}

\subsection{Fixed-price contraction}

We next show that, for a monotone price \(F\), the fixed-price LP
contracts exactly to the two-path optimization problem.

For \(F\in\mathcal A_m\), define the reflected table
\begin{equation}
  \vartheta_{qi}
  :=
  F_{m+1-q,m+1-i},
  \qquad q,i\in[m],
  \qquad
  \vartheta_{0i}:=0.
  \label{eq:bp-reflection}
\end{equation}
Then \(\vartheta\) is nondecreasing in \(q\), nonincreasing in \(i\),
and satisfies \(\vartheta_{qi}+\vartheta_{iq}=1\).  Hence
\(g(i,q):=\vartheta_{qi}\), with \(g(i,0):=0\), belongs to
\(\mathcal K_m\), and this reflection bijects
\(\mathcal A_m\) with \(\mathcal K_m\).

Put \(F_{m+1,r}:=0\) and define
\begin{equation}
  d_{\ell r}:=F_{\ell r}-F_{\ell+1,r}\ge0,
  \qquad
  \Delta_{qi}:=\vartheta_{qi}-\vartheta_{q-1,i}\ge0.
  \label{eq:bp-Delta}
\end{equation}
The two increments are related by
\(d_{\ell r}=\Delta_{m+1-\ell,m+1-r}\).

Using
\(x(\ell,r,p)=y(\ell,r,p)-y(\ell-1,r,p)\),
summation by parts gives
\begin{equation}
  \frac2m\sum_{\ell,r,p}
  F_{\ell r}x(\ell,r,p)
  =
  \frac2m\sum_{\ell,r,p}
  d_{\ell r}y(\ell,r,p).
  \label{eq:bp-sbp}
\end{equation}
Scale \(v:=my\).  The constraints on \(v\) are order/equality
constraints together with covers \(v_u+v_w\ge1\).  Since
\(d_{\ell r}\ge0\), coordinatewise truncation
\(v\mapsto\min\{v,1\}\) preserves feasibility and cannot increase the
cost, so we may impose \(0\le v\le1\).

Duplicate every variable into copies \(A\) and \(B\), duplicate every
order and equality constraint, and replace every cover \(v_u+v_w\ge1\) by
\[
  A_u+B_w\ge1,
  \qquad
  B_u+A_w\ge1.
\]
We call these the \emph{crossed-cover constraints}.  Giving each copy
half of the original objective coefficient preserves the optimum:
\(v\) gives \(A=B=v\), while a feasible pair gives
\(v=(A+B)/2\).

After setting \(Z:=1-B\), all order and crossed-cover constraints are
directed implications between variables in \([0,1]\).  Their matrix
is a directed node--arc incidence matrix with unit bound rows and is
totally unimodular.  Since all right-hand sides are integral, the entire bounded
doubled-cover polytope is integral.  Hence every feasible point is a
finite convex combination of binary feasible points; in particular,
a binary optimum exists.  Let
\(\mathfrak C_m^{(2)}\) denote the binary feasible set; we call its
elements \emph{binary doubled covers}.  Equation~\eqref{eq:bp-sbp}
becomes
\begin{equation}
  m^2\mathscr R_m(F)
  =
  \min_{(A,B)\in\mathfrak C_m^{(2)}}
  \sum_{\ell,r,p}d_{\ell r}
  \bigl(A(\ell,r,p)+B(\ell,r,p)\bigr).
  \label{eq:bp-cover-value}
\end{equation}

We now describe the structure of a binary doubled cover.  Each
variable \(X(\ell,r,p)\), \(X\in\{A,B\}\), is called a
\emph{state}.  A state is \emph{occupied} if it equals \(1\), and
the \(X\)-row \((\ell,r)\) is \emph{full} if
\(X(\ell,r,p)=1\) for every \(p\in[m]\).

Fix the passive rank \(r\) and one copy, and suppress both from the
notation.  Constraint~\eqref{eq:MY-P3} identifies all states
\((\ell,p)\) with \(p>\ell\); call their common value
\(\tau_\ell\), and write \(\xi_{\ell p}\) for \(p\le\ell\).
For two states \(u,v\), write \(u\preceq v\) if the order constraints
imply \(u\le v\).  The partial order is generated by
\begin{align}
  \xi_{\ell p}&\preceq\xi_{\ell+1,p},&
  \xi_{\ell p}&\preceq\xi_{\ell+1,p+1}
  &&(p\le\ell<m),\notag\\
  \tau_\ell&\preceq u
  &&\text{for every state \(u\) in row \(\ell+1\)}.
  \label{eq:bp-state-poset}
\end{align}
For a set of states \(S\), let \(\operatorname{cl}(S)\) be its
upward closure under \(\preceq\).  In particular,
\begin{equation}
  \operatorname{cl}(\{\xi_{kk}\})
  \cap\{\text{states in row }\ell\}
  =
  \{\xi_{\ell p}:k\le p\le\ell\}
  \qquad(\ell\ge k).
  \label{eq:bp-diagonal-cone}
\end{equation}
We call \(\xi_{kk}\) a \emph{diagonal generator} and its closure a
\emph{diagonal cone}.  By
\eqref{eq:bp-diagonal-cone}, a diagonal cone contains no tail state;
in contrast, the closure of a full row contains every state in every
later row.

For binary \(A,B\), define
\[
  \overline A_{\ell r}
  :=
  \min_{p\in[m]}A(\ell,r,p),
  \qquad
  \overline B_{\ell r}
  :=
  \min_{p\in[m]}B(\ell,r,p).
\]
Thus \(\overline A_{\ell r}=1\) exactly when the \(A\)-row
\((\ell,r)\) is full.  The crossed-cover constraints are equivalent
to
\begin{equation}
  A(\ell,r,\ell)\vee\overline B_{r\ell}=1,
  \qquad
  B(\ell,r,\ell)\vee\overline A_{r\ell}=1.
  \label{eq:bp-crossed-cover}
\end{equation}

For each passive rank \(r\), define the generator sets
\[
\begin{aligned}
  S_A(r)
  &:=
  \bigcup_{\ell:\,\overline A_{\ell r}=1}
  \{(\ell,p):p\in[m]\}
  \;\cup\;
  \{(\ell,\ell):\overline B_{r\ell}=0\},\\
  S_B(r)
  &:=
  \bigcup_{\ell:\,\overline B_{\ell r}=1}
  \{(\ell,p):p\in[m]\}
  \;\cup\;
  \{(\ell,\ell):\overline A_{r\ell}=0\}.
\end{aligned}
\]
Replace, for each \(r\), the occupied states in the \(A\)- and
\(B\)-copies by
\(\operatorname{cl}(S_A(r))\) and
\(\operatorname{cl}(S_B(r))\), respectively.  Every generator was
already occupied in the original cover by
\eqref{eq:bp-crossed-cover}, and the original occupied sets were
upward closed.  Hence the new cover is componentwise contained in
the original one.  It remains feasible and has the same full-row
generators.  We call such a cover \emph{canonical}.

Reflect the row and passive-rank coordinates by
\(t:=m+1-\ell\) and \(i:=m+1-r\).  For each \(i\), the full-row
generators form a prefix in \(t\).  Let \(a_i\) and \(b_i\) be their
respective prefix lengths:
\begin{equation}
\begin{aligned}
  t\le a_i
  &\iff
  \text{the reflected \(A\)-row \((t,i)\) is a full-row generator},\\
  t\le b_i
  &\iff
  \text{the reflected \(B\)-row \((t,i)\) is a full-row generator}.
\end{aligned}
  \label{eq:bp-generator-paths}
\end{equation}
We call these rows \emph{declared full}; a row is
\emph{physically full} if all of its states are occupied after taking
the order closure.

For \(c\in\{0,\ldots,m\}^m\), define
\begin{equation}
  z_i(c)
  :=
  \max\bigl(\{j\in[m]:c_j<i\}\cup\{0\}\bigr).
  \label{eq:bp-z}
\end{equation}
If \(b_j<i\), then the reflected \(B\)-row \((i,j)\) is not declared
full, and \eqref{eq:bp-crossed-cover} forces the corresponding
\(A\)-diagonal state to equal \(1\).  Hence \(z_i(b)\) is the largest
position of an \(A\)-diagonal generator forced by the \(B\)-copy.

Let
\[
  \mathcal R^A_{ti}(a,b)
  :=
  \left\{
    p\in[m]:
    A(m+1-t,m+1-i,p)=1
  \right\}.
\]
Equation~\eqref{eq:bp-diagonal-cone} gives the exact row structure
\begin{equation}
  \mathcal R^A_{ti}(a,b)
  =
  \begin{cases}
    [m],
      &t\le a_i,\\[1mm]
    \{m+1-z_i(b),\ldots,m+1-t\},
      &a_i<t\le z_i(b),\\[1mm]
    \varnothing,
      &t>\max\{a_i,z_i(b)\}.
  \end{cases}
  \label{eq:bp-row-set}
\end{equation}
The \(B\)-row formula is obtained by interchanging \(a\) and \(b\).
In particular,
\begin{equation}
  C^A_{ti}(a,b)
  :=
  |\mathcal R^A_{ti}(a,b)|
  =
  m\one_{\{t\le a_i\}}
  +\one_{\{t>a_i\}}\bigl(z_i(b)-t+1\bigr)^+.
  \label{eq:bp-row-count}
\end{equation}

Since \(d_{\ell r}=\Delta_{m+1-\ell,m+1-r}\), the cost of the
canonical cover is
\begin{equation}
  \mathcal C_\vartheta(a,b)
  =
  \sum_{i=1}^m
  \left[
    H_i(a_i,z_i(b))
    +H_i(b_i,z_i(a))
  \right],
  \label{eq:bp-canonical-cost}
\end{equation}
where
\begin{equation}
  H_i(x,z)
  :=
  m\vartheta_{xi}
  +\sum_{t=x+1}^{z}(z-t+1)\Delta_{ti},
  \label{eq:bp-H}
\end{equation}
and empty sums are zero.

Conversely, given any pair \(a,b\in\{0,\ldots,m\}^m\), declare the
rows \(t\le a_i\) and \(t\le b_i\), add the diagonal generators
forced by \eqref{eq:bp-crossed-cover}, and take upward closures.
The resulting pair satisfies the order and crossed-cover constraints,
hence is a feasible canonical doubled cover with
\eqref{eq:bp-row-set}--\eqref{eq:bp-canonical-cost}.

We now identify nondecreasing generator sequences with the active
coordinates of the paths in \(\mathcal B_m\).  Let
\begin{equation}
  \mathcal G_m
  :=
  \{a=(a_1,\ldots,a_m)\in\{0,\ldots,m\}^m:
    a_1\le\cdots\le a_m\}.
  \label{eq:bp-G}
\end{equation}
For \(a\in\mathcal G_m\), let
\(\boldsymbol a\in\mathcal B_m\) be given by
\(\boldsymbol a_{i-1}=a_i\) for \(i\in[m]\) and
\(\boldsymbol a_m=m\).  Put
\[
  L_a(i)
  :=
  |\{j\in[m]:a_j\ge i\}|.
\]
Then \(L_a(i)=m\ell_{\boldsymbol a}(i)\).
Writing
\(\Phi_\vartheta(a,b):=\Phi_g(\boldsymbol a,\boldsymbol b)\),
equation~\eqref{eq:Phi} becomes
\begin{equation}
\begin{aligned}
  m^2\Phi_\vartheta(a,b)
  ={}&m^2-\sum_i(a_i+b_i)\\
  &+\sum_i
  \left[
    (a_i+L_b(i))\vartheta_{a_i,i}
    +(b_i+L_a(i))\vartheta_{b_i,i}
  \right].
  \label{eq:bp-compact-Phi}
\end{aligned}
\end{equation}

For \(a,b\in\mathcal G_m\), define the transpose-overlap set
\begin{equation}
  \mathcal O(a,b)
  :=
  \{(i,t)\in[m]^2:t\le a_i,\ i\le b_t\}.
  \label{eq:bp-overlap-free}
\end{equation}
We call \((a,b)\) \emph{transpose-overlap-free} if
\(\mathcal O(a,b)=\varnothing\).

\begin{lemma}[Canonical reduction and cost identity]
\label{lem:bp-canonical-reduction}
Every canonical doubled cover has a componentwise contained canonical
representative whose generator sequences belong to
\(\mathcal G_m\) and are transpose-overlap-free.  Moreover, a
minimizer of \(\Phi_\vartheta\) over \(\mathcal G_m^2\) may be chosen
transpose-overlap-free, and every transpose-overlap-free
\(a,b\in\mathcal G_m\) satisfies
\[
  \mathcal C_\vartheta(a,b)
  =
  m^2\Phi_\vartheta(a,b).
\]
\end{lemma}

\begin{proof}
\emph{Monotone generators.}
For an arbitrary sequence \(a\), define
\[
  \widehat a_i:=\min_{j\ge i}a_j.
\]
Then \(\widehat a\) is nondecreasing, and
\(z_q(\widehat a)=z_q(a)\) for every \(q\in[m]\).  Indeed, if
\(z=z_q(a)\), then \(\widehat a_i<q\) holds exactly for \(i\le z\).
Thus replacing \(a\) by \(\widehat a\) leaves the largest diagonal
cone forced at every height unchanged, while it can only remove
\(A\)-row generators and their upward closures.  The resulting cover
is therefore componentwise contained in the original one.  Apply the
same operation to \(b\).

\emph{Overlap removal.}
Suppose \((i_0,t)\in\mathcal O(a,b)\), and put
\(s:=a_{i_0}\ge t\).  Since \(b\) is nondecreasing,
\(b_s\ge b_t\ge i_0\).  Let \(i\) be the first index in the plateau
\(a_i=s\), and replace \(a_i=s\) by \(s-1\).
Monotonicity is preserved.

At the cover level, the reflected \(A\)-row \((s,i)\) ceases to be
declared full.  The only additional diagonal requirement that could
arise from \eqref{eq:bp-crossed-cover} lies in the transposed
\(B\)-row \((i,s)\).  But \(b_s\ge i\), so that row is already
declared full.  Hence no new occupied state is introduced, while one
\(A\)-row generator is removed.  After taking closures again, the new
canonical cover is componentwise contained in the old one.  No new
overlap is created because only \(a_i\) decreases.  Repetition
terminates because \(\sum_i a_i\) decreases.

For the compact path objective, the same update changes only \(a_i\)
and \(L_a(s)\), the latter decreasing by one.  The change in
\(m^2\Phi_\vartheta\), new value minus old value, is
\begin{align*}
  &1-\vartheta_{s-1,i}
   -(s+L_b(i))\Delta_{si}
   -\vartheta_{b_s,s}\\
  &\qquad
   =\vartheta_{s,b_s}-\vartheta_{s-1,i}
    -(s+L_b(i))\Delta_{si}\\
  &\qquad
   \le(1-s-L_b(i))\Delta_{si}
   \le0.
\end{align*}
The equality uses complementarity, while
\(b_s\ge i\) and monotonicity in the second coordinate give
\[
  \vartheta_{s,b_s}
  \le\vartheta_{s,i}
  =\vartheta_{s-1,i}+\Delta_{si}.
\]
Thus a minimizing pair may be chosen transpose-overlap-free.

\emph{Cost identity.}
For nondecreasing \(a,b\),
\[
  z_i(b)=m-L_b(i),
  \qquad
  z_i(a)=m-L_a(i).
\]
The condition \(\mathcal O(a,b)=\varnothing\) is equivalent to
\[
  a_i\le m-L_b(i),
  \qquad
  b_i\le m-L_a(i)
  \qquad(i\in[m]).
\]
Summation by parts in \eqref{eq:bp-H} gives
\begin{align}
  \mathcal C_\vartheta(a,b)
  =\sum_i\Biggl[
  &(a_i+L_b(i))\vartheta_{a_i,i}
   +\sum_{q=a_i+1}^{m-L_b(i)}\vartheta_{qi}\notag\\
  &+(b_i+L_a(i))\vartheta_{b_i,i}
   +\sum_{q=b_i+1}^{m-L_a(i)}\vartheta_{qi}
  \Biggr].
  \label{eq:bp-cost-telescoped}
\end{align}

Let
\[
  \mathcal Q
  :=
  \{(q,i)\in[m]^2:
    a_i<q\le m-L_b(i)\}.
\]
Since \(b\) is nondecreasing,
\[
  (q,i)\in\mathcal Q
  \iff
  a_i<q\ \text{ and }\ b_q<i.
\]
Thus the two residual sums in
\eqref{eq:bp-cost-telescoped} are indexed by
\(\mathcal Q\) and \(\mathcal Q^{\mathsf T}\), respectively.
Complementarity gives
\[
  \sum_{(q,i)\in\mathcal Q}\vartheta_{qi}
  +
  \sum_{(q,i)\in\mathcal Q^{\mathsf T}}\vartheta_{qi}
  =
  |\mathcal Q|.
\]
Finally, double counting the pairs
\(\{(i,j):i\le b_j\}\) gives
\(\sum_iL_b(i)=\sum_i b_i\), and therefore
\[
  |\mathcal Q|
  =
  \sum_i(m-L_b(i)-a_i)
  =
  m^2-\sum_i(a_i+b_i).
\]
Comparison with \eqref{eq:bp-compact-Phi} proves the cost identity.
\end{proof}

\begin{theorem}[Fixed-price contraction]
\label{thm:appendix-fixed-price}
For every \(F\in\mathcal A_m\), define
\(g(i,q):=F_{m+1-q,m+1-i}\) for \(i,q\in[m]\), with
\(g(i,0):=0\).  Then
\[
  \mathscr R_m(F)
  =
  \min_{\boldsymbol a,\boldsymbol b\in\mathcal B_m}
  \Phi_g(\boldsymbol a,\boldsymbol b).
\]
\end{theorem}

\begin{proof}
Let \(\vartheta\) be the reflected table associated with \(F\).
Start with a minimum doubled cover in
\eqref{eq:bp-cover-value} and replace it by its canonical reduction.
Lemma~\ref{lem:bp-canonical-reduction} gives a componentwise
contained canonical cover with nondecreasing,
transpose-overlap-free generators \(a,b\in\mathcal G_m\).
Because all \(d_{\ell r}\ge0\), its cost cannot increase, and the
cost identity gives
\[
  m^2
  \min_{a,b\in\mathcal G_m}\Phi_\vartheta(a,b)
  \le
  m^2\mathscr R_m(F).
\]

Conversely, take a minimizer of \(\Phi_\vartheta\) over
\(\mathcal G_m^2\) and use
Lemma~\ref{lem:bp-canonical-reduction} to remove all transpose
overlaps without increasing its value.  The canonical construction
above gives a feasible doubled cover of cost
\(m^2\Phi_\vartheta(a,b)\).  Equation~\eqref{eq:bp-cover-value}
gives the reverse inequality.  The bijection between
\(\mathcal G_m\) and \(\mathcal B_m\) completes the proof.
\end{proof}

\subsection{Minimax formulation of monotone prices}

We next evaluate the maximization over monotone price tables
explicitly.  The resulting expression consists of the total aggregate
mass together with an explicit penalty \(\sigma(\delta)\) determined
by the differences between transposed entries of the aggregate
matrix.

\begin{lemma}[Compact price domain]
\label{lem:bp-finite-price}
For an arbitrary real table \(F\),
\[
  \mathscr R_m(F)>-\infty
  \iff
  F_{\ell r}\ge0
  \quad(\ell,r\in[m]).
\]
Consequently, for \(F\in\mathcal D_m\), finiteness is equivalent to
\(0\le F\le1\).  For every \(F\ge0\), the infimum defining
\(\mathscr R_m(F)\) is unchanged if one imposes
\(0\le y(\ell,r,p)\le1/m\).
\end{lemma}

\begin{proof}
For any nonnegative matrix \(N\in\R_+^{m\times m}\), define
\[
  x^{(N)}(\ell,r,p):=\frac{N_{\ell r}}m,
  \qquad
  y^{(N)}(\ell,r,p)
  :=
  \frac1m\sum_{h=1}^{\ell}N_{hr}.
\]
The array \(y^{(N)}\) is independent of \(p\) and nondecreasing in
\(\ell\), so \((x^{(N)},y^{(N)})\) satisfies the homogeneous forms of
\eqref{eq:MY-P1}--\eqref{eq:MY-P4}.  Hence, for every feasible
\((x,y)\) and every \(s\ge0\), the point
\((x+sx^{(N)},y+sy^{(N)})\) is feasible, and
\(\sum_p x(\ell,r,p)\) increases by \(sN_{\ell r}\).
If some \(F_{\ell r}<0\), choose \(N=E^{\ell r}\), where \(E^{\ell r}\) denotes the \(m\times m\) matrix unit with a \(1\) in position \((\ell,r)\) and zeros elsewhere, and let
\(s\to\infty\); then
\(\mathscr R_m(F)=-\infty\).  Conversely, if \(F\ge0\), the objective
is nonnegative.  Since \(F_{\ell r}+F_{r\ell}=1\) on
\(\mathcal D_m\), entrywise nonnegativity is equivalent to
\(0\le F\le1\).

Let \(T(z):=\min\{z,1/m\}\).  For the compactness statement, replace
a feasible \(y\) by its coordinatewise truncation
\(\widehat y:=T\circ y\), set \(\widehat y(0,r,p):=0\), and recover
\(\widehat x\) from consecutive differences in \(\ell\).  Thus
\(\widehat x(\ell,r,p):=\widehat y(\ell,r,p)-
\widehat y(\ell-1,r,p)\).  The order relations and equalities are
preserved.  Also, if two nonnegative numbers have sum at least
\(1/m\), then truncating both at \(1/m\) leaves their sum at least
\(1/m\).  Hence every crossed-coverage constraint is preserved.
Since \(y(\ell,r,p)\) is nondecreasing in \(\ell\) and \(T\) is
nondecreasing and \(1\)-Lipschitz,
\(0\le\widehat x\le x\).  Hence, for \(F\ge0\), every feasible point
has a capped feasible replacement of no larger cost.  Since the
capped region is a subset of the original feasible region, the two
infima are equal.  The capped feasible region is closed and bounded
in a finite-dimensional space, and hence compact.
\end{proof}

Let
\(\widehat{\mathcal C}_m:=\{(x,y)\in\mathcal C_m:
0\le y(\ell,r,p)\le 1/m\text{ for all }\ell,r,p\}\).
By Lemma~\ref{lem:bp-finite-price}, restricting the inner minimization to
\(\widehat{\mathcal C}_m\) does not change its value for \(F\ge0\). And set
\begin{equation}
  v:=my,
  \qquad
  \xi:=mx,
  \qquad
  M_{\ell r}
  :=
  \sum_{p=1}^m\xi(\ell,r,p)
  =
  mW_{\ell r}.
  \label{eq:bp-scaled-M}
\end{equation}
Thus \(M_{\ell r}\) is the scaled total mass assigned to the ordered
pair \((\ell,r)\), after summing over the state index \(p\).  In this
notation the objective of \(\polyLPp(m)\) is
\(m^{-2}\sum_{\ell,r}M_{\ell r}\).

For \(1\le\ell<r\le m\), identify the ordered pair with the interval
\(I=[\ell,r]\subseteq[m]\), and define
\(h_I:=2F_{\ell r}-1\) and
\(\delta_I:=M_{\ell r}-M_{r\ell}\).
Thus \(\delta_I\) measures the difference between a matrix entry and
its transposed entry; it vanishes exactly when
\(M_{\ell r}=M_{r\ell}\).
Let \(\mathcal I_m:=\{[\ell,r]:1\le\ell<r\le m\}\), ordered by
inclusion.  A family
\(\mathcal U\subseteq\mathcal I_m\) is \emph{upward closed} if
\(I\in\mathcal U\), \(J\in\mathcal I_m\), and \(I\subseteq J\) imply
\(J\in\mathcal U\).  Define
\begin{equation}
  \sigma(\delta)
  :=
  \max_{\mathcal U\subseteq\mathcal I_m:
       \,\mathcal U\text{ upward closed}}
  \sum_{I\in\mathcal U}\delta_I.
  \label{eq:bp-sigma}
\end{equation}

\begin{lemma}[Interval-order minimax]
\label{lem:bp-interval-minimax}
The map \(F\mapsto(h_I)_{I\in\mathcal I_m}\) is a bijection from
\(\mathcal A_m\) onto the order polytope
\[
  0\le h_I\le1,
  \qquad
  I\subseteq J\Longrightarrow h_I\le h_J.
\]
Moreover, if
\(V_m^A:=\max_{F\in\mathcal A_m}\mathscr R_m(F)\), then
\begin{equation}
  V_m^A
  =
  \min_{(x,y)\in\widehat{\mathcal C}_m}
  \frac1{m^2}
  \left[
    \sum_{\ell,r}M_{\ell r}
    +\sigma(\delta)
  \right].
  \label{eq:bp-A-minimax}
\end{equation}
\end{lemma}

\begin{proof}
For \(\ell<r\), price monotonicity compares
\(F_{\ell r}\) with \(F_{rr}=1/2\), hence \(h_{[\ell,r]}\ge0\).
Moving the left endpoint left or the right endpoint right can only
increase \(F_{\ell r}\), so \(I\subseteq J\) implies
\(h_I\le h_J\).
Conversely, these inequalities, together with
\(F_{r\ell}=1-F_{\ell r}\) and \(F_{\ell\ell}=1/2\), recover both
coordinate monotonicities and the bounds in
\eqref{eq:bp-A}.

Grouping each off-diagonal term with its transpose gives
\begin{equation}
  \frac2{m^2}\sum_{\ell,r}F_{\ell r}M_{\ell r}
  =
  \frac1{m^2}
  \left[
    \sum_{\ell,r}M_{\ell r}
    +\sum_{I\in\mathcal I_m}h_I\delta_I
  \right].
  \label{eq:bp-price-pairing}
\end{equation}
For every feasible \(h\),
\[
  \sum_I h_I\delta_I
  =
  \int_0^1
  \sum_{I:h_I\ge t}\delta_I\,\dd t.
\]
Every level set \(\{I:h_I\ge t\}\) is upward closed, so the last
display is at most \(\sigma(\delta)\).  Conversely, the indicator of
an upward-closed family attaining the maximum in
\eqref{eq:bp-sigma} belongs to the order polytope and attains
\(\sigma(\delta)\).  Therefore
\[
  \max_{h\text{ in the order polytope}}\sum_Ih_I\delta_I
  =
  \sigma(\delta).
\]
Both \(\mathcal A_m\) and
\(\widehat{\mathcal C}_m\) are compact and convex, and
\eqref{eq:bp-price-pairing} is bilinear.  The finite-dimensional
minimax theorem therefore allows the max and min to be interchanged,
which gives \eqref{eq:bp-A-minimax}.
\end{proof}

\subsection{Balanced completion by predecessor transport}

By Lemma~\ref{lem:bp-partial-dual} and
\(\mathcal A_m\subseteq\mathcal D_m\), we already have
\(V_m^A\le\operatorname{val}(\polyLPp(m))\).
To prove the reverse inequality, it is enough, by
\eqref{eq:bp-A-minimax}, to choose a minimizer with aggregate matrix
\(M\) and construct a point feasible for the original uncapped region
\(\mathcal C_m\), with aggregate matrix \(M'\), such that
\[
  M'_{\ell r}=M'_{r\ell}
  \qquad(\ell,r\in[m])
\]
and
\[
  \sum_{\ell,r}M'_{\ell r}
  \le
  \sum_{\ell,r}M_{\ell r}+\sigma(\delta).
\]
We construct such a balanced completion below.

We first make a normalization that preserves feasibility and adds no
occupied states.  By \eqref{eq:bp-row-set}, an \(A\)-row that is not declared
full can nevertheless be physically full only if
\begin{equation}
  t=1,\qquad a_i=0,\qquad z_i(b)=m,
  \label{eq:bp-exceptional-row}
\end{equation}
Equivalently, such a row has \(t=1\), \(a_i=0\), and \(b_m<i\).
Call such a row
\emph{exceptional}; define exceptional \(B\)-rows symmetrically.

Starting from a canonical transpose-overlap-free binary doubled cover, every
exceptional row can be eliminated without adding occupied states.
Indeed, if \(A\) has an exceptional row, let
\(i_\star:=\max\{i:a_i=0\}\).  Then \(i_\star>b_m\), so the reflected
row \((1,i_\star)\) is physically full.  Replace
\(a_{i_\star}=0\) by \(1\).  Since \(i_\star\) is the last zero of
the nondecreasing sequence \(a\), the sequence remains nondecreasing.
This change only declares an already physically full row, and hence
adds no occupied \(A\)-state. Recompute the \(B\)-copy using the
canonical generator rule. Relative to the old generator set, this
removes exactly the diagonal generator forced by the former
non-fullness of the newly declared \(A\)-row, while all other full-row
and forced-diagonal generators remain present.

Thus every crossed-cover constraint other than the one corresponding
to the deleted generator retains its previous canonical witness, and
that remaining constraint is now witnessed by the newly declared full
\(A\)-row. Consequently, the resulting pair is again a feasible
canonical doubled cover. Its \(A\)-occupied set is unchanged, whereas
its \(B\)-occupied set can only shrink.  The only possible new transpose overlap would involve
the newly declared reflected row \((1,i_\star)\), but such an overlap
would require \(b_1\ge i_\star\), contrary to
\(b_1\le b_m<i_\star\).  At each step the integer
\(\sum_i(a_i+b_i)\) increases by one and is at most \(2m^2\), so
alternating between the two copies terminates.

We call the resulting canonical doubled cover \emph{saturated}; in a
saturated cover, every physically full row is declared full.

\begin{lemma}[Predecessor transport]
\label{lem:bp-predecessor}
Among the minimizers \((x,y)\in\widehat{\mathcal C}_m\) of
\eqref{eq:bp-A-minimax}, one may choose one whose scaled cumulative
array \(v:=my\) admits a representation
\[
  (v,v)=\sum_s\lambda_s(A^s,B^s)
\]
as a finite convex combination of saturated canonical binary doubled
covers with generator sequences \(a^s,b^s\in\mathcal G_m\). Let
\(M=M(v)\) be the aggregate matrix defined in
\eqref{eq:bp-scaled-M}.

Suppose that \(e_{krL}\ge0\), for \(1\le L<k\le m\) and
\(r\in[m]\), satisfy
\begin{equation}
  \sum_{L<k}e_{krL}\le (M_{kr}-M_{rk})^+
  \qquad (k=2,\ldots,m,\ r\in[m]).
  \label{eq:bp-transport-demands}
\end{equation}
Then there is a point \((x',y')\in\widehat{\mathcal C}_m\), with
scaled cumulative array \(v':=my'\) and aggregate matrix \(M'\), such
that
\begin{equation}
  M'_{\ell r}
  =M_{\ell r}
   -\sum_{L<\ell}e_{\ell rL}
   +\sum_{k>\ell}e_{kr\ell}
  \qquad(\ell,r\in[m]).
  \label{eq:bp-transport-outcome}
\end{equation}
In particular, taking \(e_{kr,k-1}=e\) and all other requests equal
to zero moves any \(0\le e\le(M_{kr}-M_{rk})^+\) from \(M_{kr}\)
to \(M_{k-1,r}\), while preserving feasibility and total mass.
\end{lemma}

\begin{proof}
\emph{A saturated minimizing representation.}
Choose a minimizing \((x^0,y^0)\in\widehat{\mathcal C}_m\), and put
\(v^0:=my^0\). By integrality of the bounded doubled-cover polytope,
\((v^0,v^0)\) is a finite convex combination
\(\sum_s\alpha_s(A_0^s,B_0^s)\) of binary doubled covers. Apply to
each component the canonical reduction, the monotone and
transpose-overlap normalization of Lemma~\ref{lem:bp-canonical-reduction}, and then the saturation
procedure above; denote the result by
\((\widetilde A^s,\widetilde B^s)\), and put
\(\widetilde v:=\sum_s\alpha_s
(\widetilde A^s+\widetilde B^s)/2\).

Every transformation above produces another feasible binary doubled
cover. Therefore their convex combination is feasible for the
doubled polytope. Averaging the two copies shows that $\widetilde v
  $ is feasible for the capped one-copy system. Indeed, averaging the two
crossed-cover inequalities gives the corresponding one-copy cover
inequality, while the order and equality constraints are preserved
under averaging. Since every component is binary, \(0\le\widetilde
v\le1\), so the cap is also satisfied.

At the level of occupied states, each normalized component is
componentwise contained in the original one. Since
\(d_{\ell r}\ge0\), \eqref{eq:bp-sbp} therefore shows that the
fixed-price cost of \(\widetilde v\) is no larger than that of \(v^0\)
for every \(F\in\mathcal A_m\). By
\eqref{eq:bp-price-pairing} and the definition of \(\sigma\), the
objective in \eqref{eq:bp-A-minimax} is the maximum of these
fixed-price costs. Hence the objective value at \(\widetilde v\) is
no larger than that at \(v^0\). Since \(v^0\) is already a minimizer,
equality holds, and \(\widetilde v\) is again a minimizer.

Finally, replace every normalized component by it and its copy-swapped
version, each with half of the original weight. Copy swapping
preserves all the stated properties. After renaming components and
weights, and relabeling \(\widetilde v\) as \(v\), we obtain
\begin{equation}
  (v,v)=\sum_s\lambda_s(A^s,B^s),
  \qquad
  v=\sum_s\frac{\lambda_s}{2}(A^s+B^s),
  \qquad
  \lambda_s\ge0,\quad\sum_s\lambda_s=1.
  \label{eq:bp-saturated-mixture}
\end{equation}
In particular, the weighted averages of the \(A\)- and \(B\)-copies
are both \(v\).

\emph{Movable capacity in one component.}
Fix a component, suppress its index, and set
\(X(0,r,p):=0\) for \(X\in\{A,B\}\). For all
\(\ell,r\in[m]\), define its contribution to the aggregate matrix by
\begin{equation}
  W^X_{\ell r}
  :=\sum_{p=1}^m
  \bigl(X(\ell,r,p)-X(\ell-1,r,p)\bigr).
  \label{eq:bp-component-increment}
\end{equation}
Defining \(W^X_{\ell r}\) also at \(\ell=1\) will be needed when a
coordinate is transposed.

Consider the \(A\)-copy. Fix \(r\in[m]\), put
\(i:=m+1-r\), \(x:=a_i\), and \(z:=z_i(b)\), and abbreviate
\(N_s:=C^A_{si}(a,b)=|R^A_{si}(a,b)|\) for \(s\in[m]\), with
\(N_{m+1}:=0\). If \(t:=m+1-k\), then
\(W^A_{kr}=N_t-N_{t+1}\), also when \(k=1\).

Now let \(k\ge2\), and write \(q_p:=A(k,r,p)\) and
\(u_p:=A(k-1,r,p)\). Once \(q\) is fixed, the coordinatewise largest
feasible predecessor is
\begin{equation}
  \overline u_p=
  \begin{cases}
    \min\{q_p,q_{p+1}\},&p<k,\\
    \min_{j\in[m]}q_j,&p\ge k.
  \end{cases}
  \label{eq:bp-max-predecessor}
\end{equation}
Let \(w=(w_p)_{p\in[m]}\in\{0,1\}^m\) be any feasible replacement
for the predecessor row \(u=A(k-1,r,\cdot)\), with all other states
fixed. Then
\[
  w_p\le q_p,
  \qquad
  w_p\le q_{p+1}
  \quad (p<k),
\]
by the two cumulative-order relations. The tail equalities force all
coordinates \(w_p\), \(p\ge k\), to have a common value, and the order
relations bound this value by \(\min_{j\in[m]}q_j\). Thus every
feasible predecessor is coordinatewise bounded by the vector
\(\overline u\) displayed in
\eqref{eq:bp-max-predecessor}.

Conversely, \(\overline u\) satisfies the tail equalities and all
order constraints linking rows \(k-1\) and \(k\). Increasing the
current predecessor \(u\) toward \(\overline u\) only relaxes the
constraints coming from row \(k-2\) and the crossed-cover
constraints. Hence \(\overline u\) is indeed the coordinatewise
largest feasible predecessor. Consequently,
$
  \kappa^A_{kr}
  :=\sum_{p=1}^m(\overline u_p-u_p)$
is exactly the one-step movable capacity in this copy. Also put
\(U^A_{kr}:=\one_{\{x<t\le z\}}\).

Equations~\eqref{eq:bp-row-set}--\eqref{eq:bp-row-count} and
\eqref{eq:bp-max-predecessor} give the following mutually exclusive
cases:
\begin{equation}
\begin{array}{c|ccc}
  \text{condition}&W^A_{kr}&\kappa^A_{kr}&U^A_{kr}\\ \hline
  t<x&0&0&0\\
  t=x&m-(z-x)^+&m-(z-x)^+&0\\
  x<t\le z&1&0&1\\
  t>\max\{x,z\}&0&0&0.
\end{array}
\label{eq:bp-w-cases}
\end{equation}
The only point requiring more than direct substitution is the third
line. In the case \(x<t\le z\), formula
\eqref{eq:bp-row-set} gives
\[
  R^A_{ti}(a,b)
  =\{m+1-z,\ldots,m+1-t\}.
\]
This row is physically full if and only if \(z=m\) and \(t=1\).
Since \(x<t\), this is precisely the exceptional configuration
\((x,t,z)=(0,1,m)\), which is excluded by saturation. Outside this
configuration, direct substitution into
\eqref{eq:bp-max-predecessor} gives
\(\overline u=u\), and hence \(\kappa^A_{kr}=0\). Consequently
\begin{equation}
  \kappa^A_{kr}=W^A_{kr}-U^A_{kr},
  \qquad
  U^A_{kr}
  =\one_{\{a_i<t\le z_i(b)\}}
  =\one_{\{a_i<t,\ b_t<i\}}.
  \label{eq:bp-exact-movable}
\end{equation}
The last equality uses that \(b\in\mathcal G_m\): the set
\(\{j:b_j<i\}\) is a prefix, so \(t\le z_i(b)\) if and only if
\(b_t<i\).

If \(U^A_{kr}=1\), then \(b_t<i\le z_t(a)\). At the transposed
coordinate \((r,k)\) of the \(B\)-copy, the local parameters are
\(b_t\) and \(z_t(a)\). The row-count calculation in the third line
of \eqref{eq:bp-w-cases}, which remains valid at first index \(1\)
because of the convention in \eqref{eq:bp-component-increment}, gives
\(W^B_{rk}=1\). Hence \(U^A_{kr}\le W^B_{rk}\). By symmetry,
\(\kappa^B_{kr}=W^B_{kr}-U^B_{kr}\) and
\(U^B_{kr}\le W^A_{rk}\).

Restoring the component index and using
\eqref{eq:bp-saturated-mixture}, we obtain
\begin{equation}
\begin{aligned}
  \operatorname{cap}_{kr}
  &:=\sum_s\frac{\lambda_s}{2}
       \bigl(\kappa^{A,s}_{kr}+\kappa^{B,s}_{kr}\bigr)\\
  &=M_{kr}-\sum_s\frac{\lambda_s}{2}
       \bigl(U^{A,s}_{kr}+U^{B,s}_{kr}\bigr)\\
  &\ge M_{kr}-M_{rk}.
\end{aligned}
\qquad\text{Hence}\qquad
  \operatorname{cap}_{kr}\ge(M_{kr}-M_{rk})^+.
\label{eq:bp-movable-capacity}
\end{equation}

\emph{Transport within one atom.}
Fix a component \(s\), a copy \(X\in\{A,B\}\), and a column \(r\).
Call this copy--column piece an atom. Its coefficient in the second
average in \eqref{eq:bp-saturated-mixture} is
\(\omega:=\lambda_s/2\). A subatom of coefficient
\(0\le\rho\le\omega\) is realized by splitting off doubled-mixture
weight \(2\rho\); augmenting one cumulative row by \(c\) states on
that subatom raises its weighted row total by \(\rho c\).

Put \(i:=m+1-r\), and let \(N_q\) denote the occupied-state count of
reflected row \((q,i)\) in the designated copy; thus
\(N_q=C^A_{qi}(a^s,b^s)\) for \(X=A\), and
\(N_q=C^A_{qi}(b^s,a^s)\) for \(X=B\). By \eqref{eq:bp-w-cases}, this atom has positive movable capacity at most one source rank. If that rank is \(k\), put \(t:=m+1-k\).
Positive movable capacity can occur only in the case \(t=x\) of
\eqref{eq:bp-w-cases}. Hence row \(t=x\) is the last declared full
row. Since \(k\ge2\), row \(t+1\) exists; it is not declared full
and, by saturation, is not physically full. Consequently,
$c_1:=m-N_{t+1}=\kappa^X_{kr}>0$.

Fix \(L<k\), put \(D:=k-L\), and define
\(c_j:=m-N_{t+j}\) for \(j=1,\ldots,D\). The row counts in
\eqref{eq:bp-row-count} are nonincreasing, so
\begin{equation}
  0<c_1\le c_2\le\cdots\le c_D.
  \label{eq:bp-capacity-chain}
\end{equation}
For \(0<\varepsilon\le\omega c_1\), set
\(\rho_j:=\varepsilon/c_j\). Then
\(\omega\ge\rho_1\ge\cdots\ge\rho_D\), so choose nested subatoms
\(S_D\subseteq\cdots\subseteq S_1\) with coefficients \(\rho_j\).
Proceed inductively for \(j=1,\ldots,D\). On \(S_1\), reflected row
\(t\) is already full. For \(j>1\), the inclusion
\(S_j\subseteq S_{j-1}\) implies that reflected row \(t+j-1\) was
filled on \(S_j\) at the preceding step. Thus, on \(S_j\), the
successor row of \(t+j\) is full. Formula
\eqref{eq:bp-max-predecessor} therefore permits reflected row
\(t+j\) to be filled completely.

The operation adds \(c_j\) states per unit coefficient and hence
raises original cumulative row
$m+1-(t+j)=k-j$
by \(\rho_jc_j=\varepsilon\). It also preserves feasibility:
the predecessor construction preserves the order constraints, a full
row satisfies the tail equalities, and every crossed-cover inequality
can only be relaxed because all changes are from \(0\) to \(1\).

Thus the weighted cumulative-row totals for
\(L,\ldots,k-1\) each increase by \(\varepsilon\). For either the
pre- or post-modification cumulative array \(w\), write
\(T_{\ell r}(w):=\sum_{p=1}^m w(\ell,r,p)\), with
\(T_{0r}(w):=0\). The modification increases
\(T_{\ell r}\) by \(\varepsilon\) exactly for
\(L\le\ell\le k-1\), and leaves all other cumulative-row totals
unchanged. Since
\(M_{\ell r}(w)=T_{\ell r}(w)-T_{\ell-1,r}(w)\), the consecutive
differences telescope: \(M_{kr}\) decreases by \(\varepsilon\),
\(M_{Lr}\) increases by \(\varepsilon\), and every intermediate entry
is unchanged. Thus the requested mass is transported from
\(M_{kr}\) to \(M_{Lr}\), with total aggregate mass preserved.


\emph{Simultaneous realization.}
For fixed \((k,r)\), let \(\mathfrak A_{kr}\) be the finite set of
atoms with that movable source. These sets are pairwise disjoint,
because an atom has at most one movable source. For
\(\nu\in\mathfrak A_{kr}\), let \(\omega_\nu\) be its coefficient
and \(c^\nu_1\) its capacity per unit coefficient. Then
\(\operatorname{cap}_{kr}
=\sum_{\nu\in\mathfrak A_{kr}}\omega_\nu c^\nu_1\). By
\eqref{eq:bp-transport-demands} and
\eqref{eq:bp-movable-capacity}, the requests can be allocated so that
\begin{equation}
  \sum_{\nu\in\mathfrak A_{kr}}e^\nu_{krL}=e_{krL}
  \quad(L<k),
  \qquad
  \sum_{L<k}e^\nu_{krL}\le\omega_\nu c^\nu_1
  \quad(\nu\in\mathfrak A_{kr}).
  \label{eq:bp-atom-allocation}
\end{equation}
A greedy allocation from the atom capacities gives such numbers. For each atom \(\nu\),
\eqref{eq:bp-atom-allocation} gives
$
  \sum_{L<k}\frac{e^\nu_{krL}}{c^\nu_1}
  \le\omega_\nu$.
Hence the initial subatoms, one for each \(L<k\), can indeed be
chosen pairwise disjoint.

Within atom \(\nu\), choose disjoint initial subatoms of coefficients
\(e^\nu_{krL}/c^\nu_1\), one for each \(L<k\), and apply the nested
construction above. To realize all splittings in a single finite
mixture, represent a doubled component of weight \(\lambda_s\) by an
interval of length \(\lambda_s\). A subatom of coefficient \(\rho\)
uses an interval of length \(2\rho\). Take the common refinement of the finitely many initial and nested
interval partitions. Partitions belonging to different copies or
columns may overlap, but on each cell of the common refinement all
prescribed modifications are unambiguous. Operations belonging to
different copy--column pieces affect disjoint order systems, while
within a fixed piece the initial subatoms for different destinations
are disjoint. Moreover, all modifications are coordinatewise
increases, so every crossed-cover inequality remains satisfied.
Thus every cell of the common refinement carries a feasible binary
doubled cover.

Let \((\overline A,\overline B)\) be the weighted average of the
modified doubled covers and put
$
  v':=\frac{\overline A+\overline B}{2}$.
Copy swapping preserves feasibility, and therefore
$
  (v',v')
  =
  \frac12(\overline A,\overline B)
  +\frac12(\overline B,\overline A)$
is feasible for the doubled polytope. Equivalently, \(v'\) is
feasible for the capped one-copy system. Setting
$
  y':=\frac{v'}m$ and $
  x'(\ell,r,p)
  :=y'(\ell,r,p)-y'(\ell-1,r,p)$ 
with \(y'(0,r,p):=0\), gives
\((x',y')\in\widehat{\mathcal C}_m\).

Since the aggregate matrix is linear in \(v\) by
\eqref{eq:bp-scaled-M}, the atom-level changes add and give exactly
\eqref{eq:bp-transport-outcome}. Each request removes and adds the
same amount, so the total aggregate mass is preserved. This completes
the proof.
\end{proof}

\begin{theorem}[Monotone-price normalization]\label{thm:appendix-price-normalization}
\[
  \max_{F\in\mathcal A_m}\mathscr R_m(F)
  =
  \sup_{F\in\mathcal D_m}\mathscr R_m(F)
  =
  \operatorname{val}\bigl(\polyLPp(m)\bigr).
\]
\end{theorem}

\begin{proof}
Choose a capped minimizer in \eqref{eq:bp-A-minimax} with the
saturated binary-mixture representation furnished by
Lemma~\ref{lem:bp-predecessor}. Write \(M\) for its aggregate matrix
and \(\delta\) for its interval-imbalance vector, and put
\(P:=\sum_{I\in\mathcal I_m}\delta_I^+\).

\smallskip
\noindent\emph{The interval flow.}
Construct a network with source \(\mathsf s\), sink \(\mathsf t\),
and one node for each \(I\in\mathcal I_m\). Add arcs
\(\mathsf s\to I\) and \(I\to\mathsf t\) of capacities
\(\delta_I^+\) and \(\delta_I^-\), respectively, and, whenever
\(I\subsetneq J\), an arc \(I\to J\) of capacity \(P+1 \).

The cut with all interval nodes on the sink side has capacity \(P\),
so no minimum cut crosses a capacity-\((P+1)\) arc. Hence the interval
nodes \(\mathcal U\) on the source side of a minimum cut form an
upward-closed family. Conversely, every upward-closed
\(\mathcal U\) defines a cut of capacity
\begin{equation}
  \operatorname{cap}(\mathcal U)
  =
  \sum_{I\notin\mathcal U}\delta_I^+
  +
  \sum_{I\in\mathcal U}\delta_I^-
  =
  P-\sum_{I\in\mathcal U}\delta_I.
  \label{eq:bp-closure-cut}
\end{equation}
If \(\mathsf{mf}\) denotes the maximum-flow value, max-flow/min-cut
duality and the definition of \(\sigma\) give
\begin{equation}
  \sigma(\delta)=P-\mathsf{mf}.
  \label{eq:bp-closure-flow}
\end{equation}

Decompose a maximum flow into source--sink paths. Since a direct
containment arc exists for every strict inclusion, each positive-value
path may be shortened to
\(\mathsf s\to I\to J\to\mathsf t\), without changing its value or
endpoint loads, where
\[
  I=[\ell,r],\qquad
  J=[L,R],\qquad
  \delta_I>0,\qquad
  \delta_J<0,\qquad
  I\subsetneq J.
\]
The shortcut arcs have sufficient capacity because the total flow is
at most \(P<P+1\). For a positive interval \(I\), let \(p_I\) be the
flow on \(\mathsf s\to I\); for a negative interval \(J\), let \(n_J\)
be the flow on \(J\to\mathsf t\). Thus
\(\sum_Ip_I=\sum_Jn_J=\mathsf{mf}\).

\smallskip
\noindent\emph{Transporting the flow.}
A path of value \(\varepsilon\) from
\(I=[\ell,r]\) to \(J=[L,R]\) generates the requests
\begin{equation}
  M_{\ell r}\longrightarrow M_{Lr}\quad(L<\ell),
  \qquad
  M_{RL}\longrightarrow M_{rL}\quad(r<R),
  \label{eq:bp-matched-requests}
\end{equation}
each of amount \(\varepsilon\); a request whose displayed strict
inequality fails is omitted. Equivalently, these amounts are added
to \(e_{\ell,r,L}\) and \(e_{R,L,r}\), respectively.

For every negative interval \(J=[L,R]\), set
\(q_J:=\delta_J^- -n_J\ge0\), and add the unmatched request
\(M_{RL}\to M_{LL}\) of amount \(q_J\), equivalently adding \(q_J\)
to \(e_{R,L,L}\).

For a positive interval \(I=[\ell,r]\), the first-request mass
actually sourced at \(M_{\ell r}\) is at most
$
  p_I
  \le \delta_I^+
  =(M_{\ell r}-M_{r\ell})^+$.
For a negative interval \(J=[L,R]\), the second-request and unmatched
mass actually sourced at \(M_{RL}\) is at most
$
  n_J+q_J
  =\delta_J^-
  =(M_{RL}-M_{LR})^+$.
The first-request sources lie above the diagonal
(\(\ell<r\)), whereas the second and unmatched sources lie below the
diagonal (\(R>L\)). Hence the two bounds never need to be added at
the same source coordinate. These are the only source types, so
\eqref{eq:bp-transport-demands} holds.

Moreover, every nontrivial request has source rank at least \(2\):
the first has source rank \(\ell\ge2\) because \(L<\ell\); the second
has source rank \(R\ge2\) because \(r<R\); and the unmatched request
has source rank \(R\ge2\) because \(L<R\). Lemma~\ref{lem:bp-predecessor} therefore realizes all requests
simultaneously.

Let \(M^{\mathrm{tr}}\) be the resulting aggregate matrix and let
\(\delta^{\mathrm{tr}}\) be its interval imbalance. By \eqref{eq:bp-transport-outcome}, request effects superpose.
For one path, the net effect is to decrease \(\delta_I\) by
\(\varepsilon\) and increase \(\delta_J\) by \(\varepsilon\).
Indeed, when \(L<\ell\) and \(r<R\), the two changes at the
intermediate interval \([L,r]\) cancel; if \(L=\ell\) or \(r=R\),
the remaining request gives the same endpoint effect directly.

The unmatched request raises the imbalance of a negative interval
\(J\) from
\(-\delta_J^-+n_J=-q_J\) to zero and affects no other off-diagonal
imbalance. A positive interval \(I\) has residual imbalance
\(\delta_I^{\mathrm{tr}}=\delta_I^+-p_I\ge0\), and an initially
balanced interval remains balanced. Consequently,
\begin{equation}
  \delta_I^{\mathrm{tr}}\ge0
  \quad(I\in\mathcal I_m),
  \qquad
  \sum_{I\in\mathcal I_m}\delta_I^{\mathrm{tr}}
  =
  P-\mathsf{mf}
  =
  \sigma(\delta),
  \qquad
  \sum_{\ell,r}M^{\mathrm{tr}}_{\ell r}
  =
  \sum_{\ell,r}M_{\ell r}.
  \label{eq:bp-transport-residual}
\end{equation}
The last equality holds because predecessor transport only moves
aggregate mass.

\smallskip
\noindent\emph{Homogeneous completion.}
Define the nonnegative residual matrix
$
  N^{\mathrm{res}}
  :=
  \sum_{I=[\ell,r]\in\mathcal I_m}
  \delta_I^{\mathrm{tr}}E^{r\ell}$.
Use the homogeneous feasible direction constructed in the proof of
Lemma~\ref{lem:bp-finite-price}, with its matrix parameter equal to
\(N^{\mathrm{res}}/m\). By that construction, the scaled aggregate
matrix increases by exactly \(N^{\mathrm{res}}\). Adding this
direction to the transported point therefore gives
$
  M':=M^{\mathrm{tr}}+N^{\mathrm{res}}$.
For every \(I=[\ell,r]\), this adds
\(\delta_I^{\mathrm{tr}}\) only to the reverse entry
\(M^{\mathrm{tr}}_{r\ell}\), and hence eliminates the residual
imbalance. Thus \(M'=M'^{\mathsf T}\).

The homogeneous direction preserves every constraint of the
uncapped region \(\mathcal C_m\). It may violate the auxiliary cap,
but that cap is not a constraint of \(\polyLPp(m)\). Moreover,
\(M'=M'^{\mathsf T}\), together with \(M'=mW'\) from
\eqref{eq:bp-scaled-M}, gives the transpose-balance constraints of
\(\polyLPp(m)\). Hence the completed point is feasible for
\(\polyLPp(m)\).

By \eqref{eq:bp-transport-residual} and the choice of the initial
minimizer,
\begin{equation}
  \frac1{m^2}\sum_{\ell,r}M'_{\ell r}
  =
  \frac1{m^2}
  \left(
    \sum_{\ell,r}M_{\ell r}+\sigma(\delta)
  \right)
  =
  V_m^A.
  \label{eq:bp-completed-objective}
\end{equation}
Therefore
\[
  \operatorname{val}\bigl(\polyLPp(m)\bigr)
  \le
  V_m^A
  =
  \max_{F\in\mathcal A_m}\mathscr R_m(F)
  \le
  \sup_{F\in\mathcal D_m}\mathscr R_m(F)
  =
  \operatorname{val}\bigl(\polyLPp(m)\bigr).
\]
The second inequality uses
\(\mathcal A_m\subseteq\mathcal D_m\), and the last equality is
Lemma~\ref{lem:bp-partial-dual}. Hence all inequalities are equalities.
\end{proof}
\begin{corollary}[Balance--price duality]
\label{cor:bp-balance-price}
\[
  V_m
  =
  \operatorname{val}\bigl(\polyLPp(m)\bigr).
\]
\end{corollary}

\begin{proof}
By Theorem~\ref{thm:appendix-price-normalization},
\(\operatorname{val}(\polyLPp(m))
=\max_{F\in\mathcal A_m}\mathscr R_m(F)\).  The map \(F\mapsto g\),
defined by \(g(i,q):=F_{m+1-q,m+1-i}\), is a bijection from
\(\mathcal A_m\) to \(\mathcal K_m\), with inverse
\(F_{\ell r}=g(m+1-r,m+1-\ell)\).
For corresponding \(F\) and \(g\),
Theorem~\ref{thm:appendix-fixed-price} gives
\[
  \mathscr R_m(F)
  =
  \min_{\boldsymbol a,\boldsymbol b\in\mathcal B_m}
  \Phi_g(\boldsymbol a,\boldsymbol b).
\]
Therefore
\[
  \operatorname{val}\bigl(\polyLPp(m)\bigr)
  =
  \max_{g\in\mathcal K_m}
  \min_{\boldsymbol a,\boldsymbol b\in\mathcal B_m}
  \Phi_g(\boldsymbol a,\boldsymbol b)
  =
  V_m,
\]
where the last equality is the definition of \(V_m\).
\end{proof}

\end{document}